\documentclass[a4paper,cleveref,USenglish]{lipics-v2021}

\hideLIPIcs  

\AtBeginDocument{%
  }

\graphicspath{{./figures/}}
\DeclareGraphicsExtensions{.pdf,.jpeg,.png}
\usepackage{subcaption}

\usepackage{multicol}
\usepackage{amsthm}
\usepackage{amsmath}
\usepackage{mathtools}
\usepackage{algorithm}
\usepackage[noend]{algpseudocode}
\makeatletter\newcommand{\algmargin}{\the\ALG@thistlm}
\algdef{SE}[UPON]{Upon}{EndUpon}[3]{\textbf{upon} \textbf{#1} $\langle \text{#2} \mid #3\rangle$ \textbf{do}}{}
\algdef{SE}[UPONEXISTS]{UponExists}{EndUpon}[1]{\textbf{upon exists} #1 \textbf{do}}{}
\algdef{SE}[UPONTRUE]{UponTrue}{EndUpon}[1]{\textbf{upon once} #1 \textbf{do}}{}
\algdef{SE}[USES]{Uses}{EndUses}{\textbf{Uses:}}{}
\ifthenelse{\equal{\ALG@noend}{t}}%
  {\algtext*{EndUses}}
  {}%
\ifthenelse{\equal{\ALG@noend}{t}}%
  {\algtext*{EndUpon}}
  {}%
\makeatother
\algnewcommand{\parState}[1]{
    \parbox[t]{\dimexpr\linewidth-\algmargin}{\strut\hangindent=\algorithmicindent \hangafter=1 #1\strut}}

\usepackage{color}

\usepackage{booktabs}
\usepackage{tabularx}

\usepackage{environ}
\NewEnviron{interface}{
\vspace{0.5em}
\par
\begin{tikzpicture}
\node[rectangle,minimum width=0.9\textwidth] (m) {\begin{minipage}{0.85\textwidth}\BODY\end{minipage}};
\draw[dashed] (m.south west) rectangle (m.north east);
\end{tikzpicture}
}

\usepackage{tikz}
\usetikzlibrary{shapes,arrows}
\definecolor{xlightgray}{HTML}{D1DBE4}
\definecolor{xgreen}{HTML}{B9C89F}
\definecolor{xred}{HTML}{F19F9E}
\definecolor{xyellow}{HTML}{F6DA71}
\definecolor{xblue}{HTML}{ACE7EF}
\definecolor{xpurple}{HTML}{C3A4E7}

\usepackage{todonotes}

\usepackage{xspace}
\newcommand{\sysname}{\textsc{Mangrove}\xspace}

\crefname{ALC@line}{line}{lines}
\Crefname{ALC@line}{Line}{Lines}

\title{\sysname: Fast and Parallelizable State Replication for Blockchains}

\author{Anton Paramonov}{ETH Zurich}{aparamonov@ethz.ch}{}{}
\author{Yann Vonlanthen}{ETH Zurich}{yvonlanthen@ethz.ch}{}{}
\author{Quentin Kniep}{ETH Zurich}{qkniep@ethz.ch}{}{}
\author{Jakub Sliwinski}{ETH Zurich}{jsliwinski@ethz.ch}{}{}
\author{Roger Wattenhofer}{ETH Zurich}{wattenhofer@ethz.ch}{}{}

\authorrunning{A. Paramonov, Y. Vonlanthen, Q. Kniep, J. Sliwinski, R. Wattenhofer} 

\Copyright{Jane Open Access and Joan R. Public} 



\category{} 

\relatedversion{} 

\EventEditors{John Q. Open and Joan R. Access}
\EventNoEds{2}
\EventLongTitle{42nd Conference on Very Important Topics (CVIT 2016)}
\EventShortTitle{CVIT 2016}
\EventAcronym{CVIT}
\EventYear{2016}
\EventDate{December 24--27, 2016}
\EventLocation{Little Whinging, United Kingdom}
\EventLogo{}
\SeriesVolume{42}
\ArticleNo{23}

\nolinenumbers

\begin{document}

\maketitle

\begin{abstract}
  \sysname is a novel scaling approach to building blockchains with parallel smart contract support.
    Unlike in monolithic blockchains, where a single consensus mechanism determines a strict total order over all transactions,
    \sysname uses separate consensus instances per smart contract, without a global order.
    To allow multiple instances to run in parallel while ensuring that no conflicting transactions are committed, we propose a mechanism called Parallel Optimistic Agreement.
    \sysname is optimized for performance under \emph{optimistic} conditions, where there is no misbehavior and the network is synchronous.
    Under these conditions, our protocol can achieve the latency of 2 communication steps between creating and executing a transaction.
\end{abstract}

\section{Introduction}
\label{sec:introduction}

Scalability remains a major challenge for blockchain systems. Arguably, no single blockchain currently offers sufficient throughput to support traditional Web 2.0 applications in a trustless manner~\cite{dfinity_forum_post}. In the modern blockchain landscape, users are fragmented across numerous Layer-1, Layer-2, and even Layer-3 chains. This fragmentation raises concerns about the interoperability and latency of decentralized applications built on top of these chains.

Techniques like sharding, where network nodes are divided into groups to process transactions in parallel~\cite{rapidchain,adhikari2024fast}, have gained considerable attention as potential solutions to the problem.
Although sharding has been proven to enhance system performance, the techniques come with drawbacks and limitations.
Specifically, sharded systems experience significant latency (or abortion rate) when dealing with smart contracts with high contention. Furthermore, they require cross-shard agreement for transactions spanning multiple shards.

Since it has been established that consensus is not necessary for applications like payments \cite{gupta,consensus_number}, the research on consensus-less payment systems \cite{fastpay} has offered a remarkably simple alternative solution to the problem of scaling. Foregoing consensus, these systems offer a model in which every validator can parallelize processing and execution of all transactions without limitation.
In contrast, sharding protocols assume a rigid division of validators that complicates the system and limits the potential for parallelizability.
However, consensus-less systems can support only a limited range of applications, as consensus is necessary for general smart contracts.
\subparagraph*{Our Contributions.}
We address the following research question:
\begin{quote}
``Can a protocol supporting consensus exhibit the advantages of consensus-less payment systems?''
\end{quote}
and answer in the positive. We summarize our contributions in the following.

\begin{itemize}
    \item We propose the Replicated Actor Model, a novel execution model for blockchain systems based on externally owned accounts (user actors) and smart contracts (reactive actors), that makes dependencies explicit with parallelizability in mind. 
    
    \item We introduce Parallelizable Optimistic Agreement (POA), a consensus primitive that can be used to achieve consensus on the stream of incoming transactions for individual smart contracts. POA instances are designed to run in parallel, while ensuring that conflicting (i.e., double-spending) transactions cannot be committed. 
    
    
    \item The resulting system, called \sysname, combines limitless parallelization, low latency, and support for general smart contracts. In \sysname, no single validator can delay the entire system's progress, and congestion at one smart contract does not impede the rest of the system. In other words, the transaction throughput of different smart contracts can be scaled horizontally by validators. Additionally, our system achieves optimal latency in optimistic conditions, where users or validators do not misbehave, and the network is synchronous. Under these conditions, a block producer can commit a transaction to a smart contract in \emph{two} communication steps.
\end{itemize}

\section{Related Work}
\label{sec:related-work}

\subparagraph*{Consensusless Systems.}

Blockchain-based systems use consensus mechanisms as their core building block.
However, consensus tolerating Byzantine faults~\cite{byzantine_generals} is inherently slow. First blockchain systems such as Bitcoin~\cite{nakamoto2008bitcoin} and Ethereum~\cite{ethereum} are notorious for their limited performance.

However, it has been established that consensus is not necessary for many applications, such as payments~\cite{gupta,consensus_number}. Designs such as Fastpay~\cite{fastpay}, Astro~\cite{astro}, and Accept~\cite{accept} propose remarkably simple solutions that inherently parallelize the workload. Validators in these systems can easily add computational resources to process more transactions. Tonkikh et al.~\cite{cryptoconcurrency} and Bazzi et al.~\cite{bazzi2024fractional} show that even dependencies between transactions of the same issuer can be resolved without consensus. 
Other systems, such as Groundhog~\cite{groundhog}, Setchain~\cite{capretto2022setchain}, and Pod~\cite{alpos2025pod} avoid consensus by using commutative semantics. Frey et al.~\cite{frey2024process} and Sridhar et al.~\cite{sridhar2025stingray} recently introduced new consensus-free objects, inspired by Byzantine fault-tolerant CRDTs~\cite{kleppmann2022making}. \sysname is orthogonal to these efforts, as it focuses on the interplay between all types of objects (with- or without consensus). 
Albouy et al.~\cite{albouy2024asynchronous} show how a consensusless system can also provide anonymity properties in a lightweight fashion.
%

Despite their numerous advantages, consensusless systems are inherently unsuitable for many applications, as general smart contracts require consensus~\cite{alpos2021synchronization}. The CoD~\cite{consensus_on_demand} primitive alleviates this problem to a limited extent, by mixing payment system-like logic with consensus as fallback. In turn, it exhibits the problems of consensus systems, like poor parallelizability.

Sui~\cite{sui_lutris,cuttlefish} is a modern blockchain system that recognizes the mentioned problems and incorporates a consensusless component in its design. Sui relies on consensus only for complex tasks, like ordering accesses to the same shared objects. We further increase scalability, by allowing parallelism even for shared objects. In addition, \sysname outperforms Sui in the number of communication rounds needed for transaction execution.
Our proposed system shares many characteristics with Basil~\cite{suri2021basil}, namely parallel execution of non-conflicting transactions and fast commits under optimistic conditions. Contrary to Basil, in \sysname users do not have to drive progress themselves and benefit from lower latency.



\subparagraph*{Fast Byzantine Consensus.}

Martin and Alvisi~\cite{fast_byzantine_consensus} present an algorithm that achieves Byzantine Consensus in just two communication steps under optimistic conditions, specifically assuming an honest leader and a synchronous network.
They introduce a parameter $p \leq f$, which represents the number of failures supported by the fast path and show a resilience bound of $n \geq 3f + 2p + 1$. Subsequent work explores the potential and pitfalls of fast consensus, both for single-shot consensus and state machine replication~\cite{song2008bosco,abraham2018revisiting,gueta2019sbft}.
Recently, Kuznetsov et al.~\cite{revisiting_optimal_resilience} and Abraham et al.~\cite{good_case_latency} revisited this topic, pointing out that for the category of protocols where the set of proposers is a subset of the validators, a lower resilience bound of $n \geq max(3f + 1, 3f + 2p - 1)$ can be achieved.

In contrast to prior \cite{revisiting_optimal_resilience} that limits the set of proposers to a subset of validators, our protocol allows all users to act as proposers. This generality requires us to meet the more stringent resilience bound of $n \geq 3f + 2p + 1$ by Martin and Alvisi. Our protocol, \sysname, matches this bound and thus achieves optimal resilience in this setting.


Finally, Flutter~\cite{monti2024fast} avoids the communication time of sending transactions to a leader, by providing a leaderless protocol, at the cost of requiring $5f + 1$ replicas. In Flutter, clients attach bets to their transactions, which serve as timestamps to establish a total order. \sysname\ is leaderless for simple UA-UA transactions, but makes use of leaders to simplify how contention is handled at reactive actors.

\subparagraph*{Parallel Execution.}

Parallelization is a natural way to improve scalability and can be applied at various levels in blockchains.
Parallel execution~\cite{block_stm,anjana2019efficient,dickerson2017adding,neiheiser2024pythia} aims to accelerate the local execution of transactions by the validator across cores, while distributed execution~\cite{pilotfish} leverages multiple machines per validator.
Both approaches concentrate on improving local execution, whereas our work targets the elimination of the bottleneck of the single agreement mechanism.

\subparagraph*{Sharding.}

Sharding~\cite{lockless_sharding,sharper,dang2019towards,adhikari2024fast} is a technique of splitting the system state among disjoint groups of validators called shards.
As shown in~\cite{adhikari2024fast}, while sharding is efficient when a transaction only accesses a part of the system within one shard, it can result in high latency or abortion rate~\cite{chainspace} for transactions that span multiple shards, especially for highly contested actors. 
In \sysname, ``popular'' actors do not slow the progress of the whole system.
Instead, actors progress independently, ensuring that the system's overall performance remains unaffected by the contention of individual actors.

\section{Replicated Actor Model}
\label{sec:model}

Today, it is common for blockchains to define a total order over all transactions~\cite{nakamoto2008bitcoin,ethereum}.
Thus, the ensuing sequential and atomic execution of transaction bundles is often the only considered execution model. The obtained atomic composability property allows for (potentially counterintuitive) applications such as flash loans~\cite{qin2021attacking}. 

In this work, we challenge the status quo of this execution model.
To this end, we define the replicated actor model, which foregoes global sequential ordering and is more suited for parallelism.
In \cref{sec:expressiveness} we discuss the model's expressiveness and even propose an extension, which optionally allows for the reintroduction of (targeted) atomic composability.

Our model is closely related to the object model of Sui~\cite{sui_lutris}. We differentiate between the following four types of components.

\subparagraph*{Actors.}
A \textit{user actor} is associated with a digital signature key pair.
We say that a key pair (user) \emph{controls} a user actor.
\textit{Reactive actors} are analogous to smart contracts and can be thought of as a Turing machine with arbitrary state that can be changed via computations.
Actors can emit new transactions.

\subparagraph*{Objects.}
\emph{Owned objects} are objects owned by some actor, e.g. gas, tokens, or NFTs.
Every type of owned object is associated with a set of actions that can be performed over it.
These actions are specified in a global \emph{read-only object} containing the type definition.
For example, for gas objects or tokens, those actions may include splitting and merging.
Ownership of objects can be transferred.
Assume actor $A$ owns a gas object $O$ worth 10 coins and wants to transfer 2 coins to actor $B$.
Then $A$ might perform \texttt{split([8, 2])} on $O$ to receive two objects worth 8 and 2 coins respectively, and transfer the latter to $B$.

\subparagraph*{Users.} A \textit{user} is an external entity associated with a key pair,
who interacts with the system by creating and giving instructions to actors through their user actors.

\subparagraph*{Validators.} The \textit{validators} are the network participants in charge of running \sysname.
Validators participate in the broadcast and consensus algorithms and are responsible for keeping a consistent state of the system by maintaining the ownership records of each owned object and the state of actors.

\subsection{Validators}
\label{sec:validators}

We consider a set of $n$ validators, 
denoted by $\mathcal{V}$, which we assume to be known to all users and validators.
We require that at most $f$ of them are Byzantine~\cite{byzantine_generals}. We call non-Byzantine validators \emph{honest}. 
Additionally, under optimistic conditions and when less than $p$ validators are Byzantine, a transaction can be committed in two communication steps. \sysname\ assumes the participation of $n \geq 3f + 2p + 1$ validators.

Each validator maintains a dedicated state associated with each actor in the system which we call an \emph{entity}. We denote an entity of validator $V$ corresponding to an actor $A$ with $V.A$.
Different entities may be located on different machines at the validator's discretion and can communicate with each other.

Entities of different validators communicate via Outer Links, and entities within a validator communicate via Inner Links. The Inner and Outer Links implement Perfect Links~\cite{cachin2011introduction} and expose the following interface:
\begin{interface}
    - \texttt{function} $Send(m, A)$: sends message $m$ to $A$\\
    - \texttt{callback} $Deliver(m, A)$: fired upon receiving message $m$ from $A$
\end{interface}

\begin{figure}[!hb]
    \centering
    \includegraphics[width=0.99\linewidth]{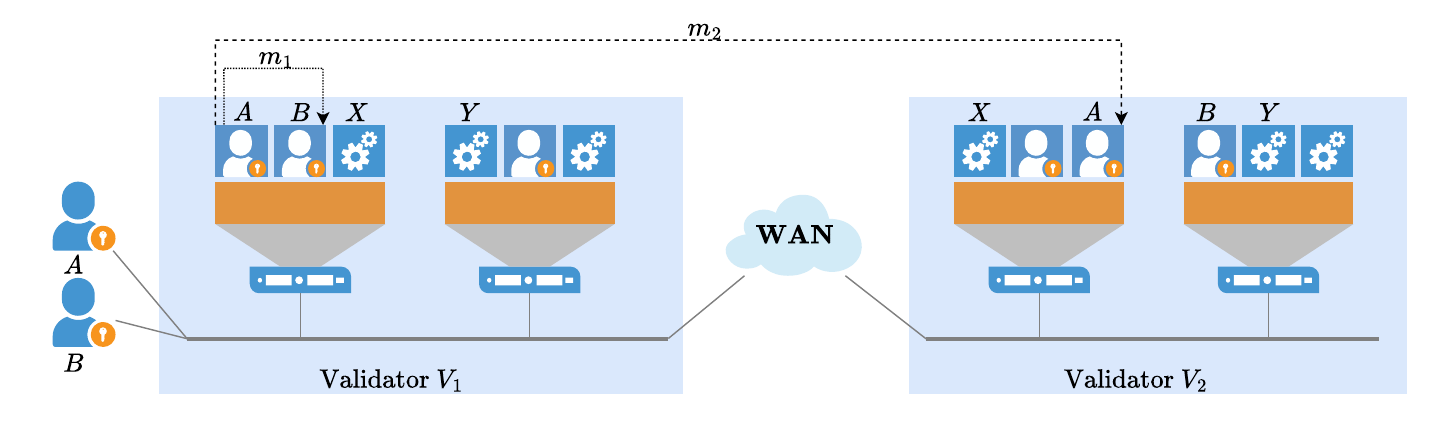}
    \caption{System entities. The validators $V_1$ and $V_2$ maintain the user actors corresponding to users $A$ and $B$, and reactive actors, e.g., $X$ and $Y$. The message $m_1$ is sent from $V_1.A$ to $V_1.B$ via Inner Links and a message $m_2$ to $V_2.A$ via Outer Links. Entities can be spread across multiple machines and organized independently by each validator.}
    \label{fig:system-entities}
\end{figure}

Apart from direct messages, validators use two agreement primitives: Parallel Optimistic Agreement (POA), which is used for transactions involving reactive actors, and Parallel Optimistic Broadcast (POB), which is used for transactions that involve only user actors. Both primitives contain a fast path, consisting of (i) a broadcast step and (ii) a single voting step.
In case of a failure (and only in case of failure), the fast path is followed by a failover mechanism (slow path) that ensures safety and liveness.
Both primitives run in consecutive instances and are described in \cref{sec:poa} and \cref{sec:pob} respectively.


Through these agreement primitives, we ensure that honest validators have a consistent view of the state of each reactive actor and the ownership of owned objects.

\subsection{Transactions}
\label{sec:transactions}

Every transaction in our system \emph{consumes} owned objects, meaning that once a transaction is executed, the objects it consumed no longer exist. Every transaction would consume a \emph{gas} object to pay for computation fees, the economics of which we leave outside the scope of this work.
Every transaction has a $Code$ field, that specifies a list of commands to perform when a transaction is being executed.
These commands can be (a) actions over owned objects a transaction consumed or created, (b) creating owned objects, (c) creating new actors, and (d) issuing new transactions.
The model differentiates between user and reactive actor transactions.

\subparagraph*{User Actor Transactions.}

Users can instruct a user actor they control to issue a transaction.
Transactions issued by the user actor are categorized based on the presence of a recipient.
If a transaction does not have a recipient, it is referred to as a UA transaction.
If it does have a recipient, which is always a reactive actor, it is referred to as a UA-RA transaction.

A UA transaction is of the form $\langle A, sn, [O_1, \ldots, O_k], Code\rangle$ where $A$, a user actor, is the sender, $sn \in \mathbb{N}$ is a sequence number and it \emph{consumes} owned objects $O_1, \ldots, O_k$ (that must be owned by $A$) and performs actions specified in $Code$ over them.
$Code$ is allowed to create new owned objects at other user actors (arbitrarily many of those), but not reactive actors.
It may also spawn new reactive actors. The reason a UA transaction can create new owned objects at user actors but not reactive actors is that the former is commutative whereas the latter requires agreement on the order.  

UA-RA transactions are of the form $\langle A, sn, X, [O_1, \ldots, O_k], Code_\mathsf{pre}, Call, Code_\mathsf{post}\rangle$ instead.
That is, they additionally include a recipient $X$ that must be a reactive actor and a $Call$ field specifying a function call to perform on $X$.
A reactive actor might issue its own transactions as a result of processing a $Call$.
$Code_\mathsf{pre}$ can operate over the consumed objects and prepare them for input into the function call.
$Code_\mathsf{post}$ can instead operate over the objects returned by the function call and, for example, decide whether and which additional transactions to spawn based on them.
The pre- and post-processing $Code$ blocks may spawn transactions of their own and can be executed in parallel to any other transaction since they do not have access to the internal state of the reactive actor.

\begin{definition}[Conflicting Transactions]
    Two user transactions $tx$ and $tx'$ with $tx \neq tx'$ and $tx.sender = tx'.sender$ are said to be \emph{conflicting} if $tx.sn = tx'.sn$.
\end{definition}
%
%
Users are responsible for issuing non-conflicting transactions with consequent sequence numbers. 
Moreover, users are also responsible for making sure that a user actor they are issuing a transaction for will eventually own the objects consumed by the transaction.
We highlight that a user failing to adhere to these requirements has no global impact on the system, but just on the actors controlled by them.

\subparagraph*{Reactive Actor Transactions.}

Unlike user actors, reactive actors only issue transactions as a potential result of executing an incoming transaction.

A reactive actor transaction (RA-RA) is always sent to another reactive actor and shares the same structure as a UA-RA transaction, except it omits the sequence number. This omission is because the order of outgoing RA-RA transactions is determined by the order of incoming transactions. The sender is the reactive actor that produced the transaction.
Upon receiving either a UA-RA or an RA-RA transaction, a reactive actor might perform computation to change its state based on the current state and $Call$ data specified in the transaction.

\subsection{Network and Computation Model}

The partial synchrony settings of Global Stabilization Time (GST) and Unknown Delta~\cite{consensus_in_partial_synchrony} constitute the gold standard of assumptions under which modern blockchains are designed to operate.
\sysname is designed to function in the GST network model, i.e., correctness is ensured even in complete asynchrony, and liveness is achieved after an arbitrary point in time called Global Stabilization Time, or GST for short.
Before GST, messages can be delayed with arbitrary delay, but every message sent at time $t$ must be delivered by $\max(t + \Delta, GST + \Delta)$.
The parameter $\Delta$ is assumed to be known and fixed. 

We assume the time required for local computations and messaging internal to a validator to be negligible.

\section{\sysname Overview}
\label{sec:system-overview}

The core principle behind \sysname is to have a dedicated agreement mechanism \emph{per actor}. In practice, agreement is achieved differently for the different types of transactions.

\begin{itemize}
    \item 
    For user actors, validators should agree on \emph{outgoing transactions} for each given sequence number. This, paired with sequential transaction execution, guarantees them a consistent view of the system~\cite{astro}.
    To this end, UA transactions are disseminated through Parallel Optimistic Broadcast (POB), and the agreement follows either from the fast path if the optimistic conditions are met, or from the failover mechanism in case they are not. 
    \item
    Instead, for reactive actors, the validators must agree on \emph{incoming transactions}. Thus, both UA-RA and RA-RA transactions go through the Parallel Optimistic Agreement (POA) mechanism of the reactive actor of the transaction recipient. The consecutive POA instances provide a total order of \emph{incoming transactions} to execute on a reactive actor.
    Importantly, since agreement is inherited from the POA properties, a dedicated agreement on the user actor can be optimistically skipped, thus also providing two step finalization latency in the fast path.
\end{itemize}

Since POB is simpler and less general than POA, we describe POA in \cref{sec:poa} and POB in \cref{sec:pob}.
The complete process for handling UA and UA-RA transactions is detailed in \cref{sec:user-actor-agreement}, while the processing of RA-RA transactions is described in \cref{sec:ra-transactions}.

\begin{figure}[ht]
    \centering
    \includegraphics[width=0.8\linewidth]{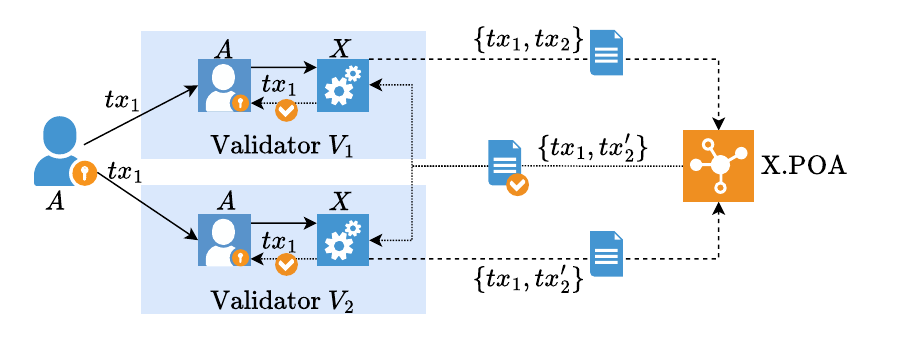}
    \caption{Mangrove UA-RA transaction processing. First, user $A$ broadcasts $tx_1$ to all $V_i.A$, who relay it to $V_i.X$ (full arrows). Then all $V_i.X$ propose $tx_1$ and other transactions they have to X.POA (dashed arrows). Finally, when $tx_1$ is included in a block decided by X.POA, all $V_i.X$ notify $V_i.A$ of the decision (dotted arrows).}
    \label{fig:ua-ra-processing}
\end{figure}
 

Both POA and POB were designed with three main goals in mind:
\begin{romanenumerate}
    \item They have a low good-case latency, that is, they terminate in two communication steps under optimistic conditions.
    \item A system running multiple instances in parallel can be sure that at most one transaction for each pair of user and sequence number is decided across all instances.
    \item A system running multiple instances in parallel does not suffer bottlenecks.
\end{romanenumerate}
    

\paragraph*{Wait-Free Locking.}

To achieve these goals, the POA and POB instances do not communicate directly but instead obtain locks via Inner Links to the user actor entities introduced in \cref{sec:model}.
These entities provide security against conflicting transactions by (locally) tying a sequence number to a transaction. To describe this process more precisely, we introduce the following notation, also used throughout the remainder of this work:

\subparagraph*{Notation.}
$U$ denotes an arbitrary user, $V.A$ denotes an arbitrary user actor entity of an arbitrary validator,
and $V.X$ denotes an arbitrary reactive actor entity of an arbitrary validator. 


User actor entities maintain the following two data structures.

\begin{definition}[Slow-Path Locked]
    Each $V.A$ maintains a map $SPLocked$, mapping sequence numbers to transactions.
    If $V.A.SPLocked[tx.sn] = tx$, we say that $V.A$ \emph{SP-locked} (slow-path locked) $tx$.
     
\end{definition}

\begin{definition}[Fast-Path Locked]
    Each $V.A$ maintains a map $FPLocked$, mapping sequence numbers to transactions.
    If $V.A.FPLocked[tx.sn] = tx$, we say that $V.A$ \emph{FP-locked} (fast-path locked) $tx$.
\end{definition}

Intuitively, SP-locking a transaction ensures that $V.X$ will never propose a conflicting transaction in the slow path, whereas FP-locking a transaction ensures $V.X$ will never vote for a proposal containing a conflicting transaction in the fast path. However, note that $V$ can FP-lock $tx$ and SP-lock $tx'$ with $tx$ and $tx'$ being conflicting. That might happen in a case $V$ received $tx$ from a user but received a lot of votes for a proposal containing $tx'$, which ``forces'' $V$ to propose $tx'$ in the slow path.

We describe the interplay between the different building blocks of \sysname\ and their correctness in \cref{sec:parallel-tx-processing} and in \cref{sec:proofs-transaction-execution}.

%

\section{Parallel Optimistic Agreement}
\label{sec:poa}

This section describes the algorithm used by validators to agree on a block $B$ of transactions to be executed at a reactive actor $X$. For every validator $V$ and reactive actor $X$, $V.X$ has a $pool$, that is, a set of transactions that $V.X$ wants to execute on $X$.
This algorithm is defined assuming a designated validator $L$, called leader.
The Parallel Optimistic Agreement (POA) primitive has an interface consisting of:

\begin{interface}
    - \texttt{function} $\textsc{Initiate}(k, pool)$: start the $k$-th instance with a transaction $pool$ \\
    - \texttt{callback} $\textsc{Decide}(k, B)$: decide a block $B$ in $k$-th instance
\end{interface}


First, in the \emph{fast path} the leader broadcasts its proposed block and all other validators cast their votes.
A validator decides on a block once it receives \textit{enough} votes, a process we refer to as a fast-path decision.
Following the fast path, and only if necessary, a \emph{slow path} (whose components are described in \cref{sec:slow path}) is initiated to ensure liveness in cases where users or the leader misbehave or the network experiences asynchrony.
We refer to a decision made in the slow path as a slow-path decision.

\begin{figure}[ht]
    \centering
    \includegraphics[width=0.8\linewidth]{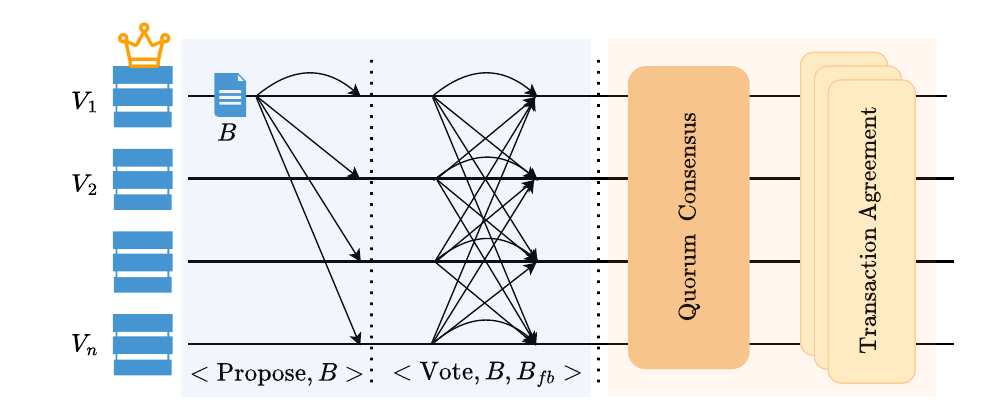}
    \caption{POA scheme. In the fast path (blue rectangle), the leader (in this case $V_1$) broadcasts their proposal $B$.
        Then, validators cast their vote on the proposal, and append their own transactions $B_{fb}$. In case the fast path fails, validators participate in the slow path (orange box), which consists of one instance of Quorum Consensus and multiple instances of Transaction Agreement.}
    \label{fig:poa_scheme}
\end{figure}

\subsection{Properties}
\label{sec:poa-properties}

A single instance of POA satisfies the following properties.

\begin{property}[Agreement]
\label{prop:agreement}
    If two honest validators decide blocks $B$ and $B'$ respectively, then $B = B'$.
\end{property}

\begin{property}[Termination]
\label{prop:termination}
    Every honest validator eventually decides a block.
\end{property}

\begin{property}[Fast Termination]
\label{prop:fast-termination}
    If $L$ is honest, at most $p$ validators misbehave, the system has reached GST, and for every UA-RA transaction $tx \in L.X.pool$ the user who issued $tx$ is honest, then all honest processes decide and stop sending messages in two communication steps. 
\end{property}

\begin{definition}[Emitted transaction]
\label{def:emit}
    We say that a transaction $tx$ is \emph{emitted} if it is either (i) a user transaction from an actor $A$ signed and broadcast at the moment where the whole system is in the state such that there exists an honest validator $V.A$ who will eventually execute a transaction with a sequence number $tx.sn - 1$ and who will eventually own all objects consumed by $tx$ or (ii) produced by a reactive actor at an honest validator.
\end{definition}

\begin{property}[Validity]
\label{prop:validity}
    (I) If a transaction $tx$ is present in the pool of all honest validators at the start of the protocol and $L$ is honest, then $tx$ is included in the decided block.
    (II) If a transaction is included in the decided block, it was emitted.
\end{property}

\begin{property}[No Conflict]
    Conflicting transactions cannot be simultaneously decided within POA instances, even if those correspond to different reactive actors.
\end{property}

\subsection{Multi-Instancing}
\label{sec:poa-multi-instancing}

\subparagraph*{Leader Oracle.}

POA is designed to run in multiple instances, and to ensure liveness of the system, it is essential to have honest leaders.
We assume all validators have access to a common $leaderOracle$ function, which given an instance number outputs a leader for that instance.
We require this function to output an honest leader infinitely often.
This can be achieved by a random choice (assuming a common random source) or by a round-robin approach.

\subparagraph*{Switching instances.}

A validator who decides in a POA instance at time $t$, broadcasts their decision along with the proof (in practice it can be a single aggregate signature). By time $\max(t + \Delta, GST + \Delta)$, every honest validator will receive a proof.
Upon receiving such a valid proof, validators rebroadcast the proof and decide.
This mechanism allows for the following property.
\begin{property}[Common Termination]
\label{prop:common-termination}
    If an honest validator decides in the $k$-th POA instance at time $t$, then every honest validator decides in this instance by at most $\max(t + \Delta, GST + \Delta)$.
\end{property}

The following property guarantees the progress of each validator. 

\begin{property}[Multi-Termination]
    For all honest validators $V$ and reactive actors $X$, $V.X$ eventually decides in the $k$-th instance of POA for $X$.
\end{property}

\subsection{Algorithm}
\label{sec:poa-algorithm}

\begin{algorithm}[]
\caption{Parallel Optimistic Agreement \textbf{on V.X} (High Level)}
\label{alg:poa-high-level-vx}
\begin{algorithmic}[1]
    \Uses\ Outer Links, Inner Links, Quorum Consensus, Transaction Agreement
    \EndUses

    \Statex

    \Function{Initiate}{k, pool} \Comment{For Leader}
        \State $proposalBlock \gets pool$
        \State Broadcast via Outer Links $proposalBlock$
    \EndFunction

    \Statex

    \UponTrue{$time \leq 3\Delta$ \textbf{and} received proposal $B$} 
        \ForAll{$tx: \text{UA-RA transaction} \in B$} \Comment{In parallel}
            \State $A \gets tx.sender$
            \State Request $V.A$ via Inner Links to FP-lock $tx$
        \EndFor
    \EndUpon

    \Statex

    \UponTrue{all $V.A$ responded to an FP-lock request} 
        \If{all FP-locks are successful} $Vote \gets B$ \Else\ $Vote \gets \bot$ \EndIf
        \State $fallBackBlock \gets pool$
        \State Broadcast via Outer Links $\langle Vote, fallBackBlock\rangle$
    \EndUpon
    
    \Statex

    \UponTrue{exists a block $B$ with at least $n - p$ votes}
        \State Decide $B$ in POA
    \EndUpon

    \Statex

    \UponTrue{$Time > 3\Delta$ \textbf{and} exists block $B'$ with at least $n - p -2f$ votes} 
        \ForAll{$tx: \text{UA-RA transaction} \in B'$} \Comment{In parallel}
            \State $A \gets tx.sender$
            \State Request $V.A$ via Inner Links to SP-lock $tx$
        \EndFor
    \EndUpon

    \Statex

    \UponTrue{all $V.A$ responded to an SP-lock request \textbf{and} all SP-locks succeeded} 
        \State Propose $B'$ to Quorum Consensus
    \EndUpon

    \Statex

    \UponTrue{$Time > 3\Delta$ \textbf{and} (received $n - f$ votes) \textbf{and} (there is no block with $n - p - 2f$ votes \textbf{or} not all SP-Locks succeeded)} 
        \State $B'' \gets$ all transactions present in at least $n - 2f$ fallback blocks
        \State Attempt to (analogously) SP-lock all UA-RA transactions in $B''$
        \State $B'' \gets B'' \setminus \{\text{transactions that failed to SP-lock}\}$
        \State Propose $B''$ to Quorum Consensus
    \EndUpon

    \Statex

    \Upon{event}{decide in Quorum Consensus}{B_{qc}}
        \ForAll{$tx: \text{UA-RA transaction} \in B_{qc}$} \Comment{In parallel}
            \State $A \gets tx.sender$
            \State Request $V.A$ via Inner Links to propose $tx$ to Transaction Agreement
        \EndFor
    \EndUpon

    \Statex

    \UponTrue{all $V.A$ responded to an Transaction Agreement request} 
        \State $Fails \gets$ all UA-RA transactions from $B_{qc}$ which $V.A$ did not decide in Transaction Agreement
        \State $B_{POA} \gets B_{qc} \setminus Fails$
        \State Decide $B_{POA}$ in POA
    \EndUpon
    
\end{algorithmic}
\end{algorithm}

\begin{algorithm}[h]
\caption{Parallel Optimistic Agreement \textbf{on V.A} (High Level)}
\label{alg:poa-high-level-va}
\begin{algorithmic}[1]

    \Upon{event}{FP-lock request}{tx, X} 
        \If{a conflicting transaction is FP-locked} respond \textit{fail} to $V.X$ \EndIf
        \If{execution preconditions for $tx$ are not met} respond \textit{fail} to $V.X$ \EndIf
        \State Respond \textit{success} to $V.X$
    \EndUpon

    \Statex

    \Upon{event}{SP-lock request}{tx, X} 
        \If{a conflicting transaction is SP-locked} respond \textit{fail} to $V.X$ \EndIf
        \State Respond \textit{success} to $V.X$
    \EndUpon

    \Statex

    \Upon{event}{Transaction Agreement request}{tx, X} 
        \State Try to SP-Lock $tx$
        \State Propose an SP-Locked transaction with sequence number $tx.sn$ to Transaction Agreement instance $tx.sn$ of $A$
    \EndUpon

    \Statex

    \Upon{event}{decide in Transaction Agreement}{tx', sn}
        \If{$tx' = tx$} respond ``Transaction Agreement $sn$ Success'' to $V.X$
        \Else\ respond ``Transaction Agreement $sn$ Failure'' to $V.X$
        \EndIf
    \EndUpon

\end{algorithmic}
\end{algorithm}

At the start of the protocol, $L.X$ forms a block consisting of all $tx \in L.X.pool$ and broadcasts it.
Upon receiving a block $B$ an honest validator $V.X$ broadcasts its vote plus its own block (called a \emph{fallback block}) consisting of all transactions $tx \in V.X.pool$.
The vote is for $B$ in case (i) for every UA-RA transaction $tx \in B$, $tx$ is correctly signed, (ii) $V.A$ manages to FP-lock $tx$, ensuring there is no conflicting transaction and $tx$ can be executed and (iii) $V$ emitted every RA-RA transaction from $B$.
Otherwise, $V.X$ broadcasts a negative vote.

If at any time a validator $V.X$ receives $n - p$ votes for some block $B$, $V.X$ decides $B$, we call it a fast-path decision. Upon fast-path deciding, a validator broadcasts a proof that $B$ can be safely decided (in practice, this can be an aggregated threshold signature \cite{shoup2000practical} of $n - p$ votes).
In case a validator does not receive $n - p$ votes in $3\Delta$ time or those votes are for different blocks, they wait until they have $n - f$ votes and propose in the slow path. 

The slow path consists of two steps: Quorum Consensus (described in detail in \ref{sec:quorum-consensus}), followed by the parallel Transaction Agreements (described in \ref{sec:transaction agreement}) for every UA-RA transaction decided in Quorum Consensus. 
A validator proposes to Quorum Consensus according to the following cases.

\begin{romanenumerate}
\item There was some block $B$ for which $V.X$ received at least $n - p -2f$ votes. In this case, it is possible that some other validator received $n - p$ votes for $B$, so $V.X$ attempts to SP-lock all UA-RA transactions in $B$, and if it succeeds, proposes $B$ to Quorum Consensus.
    
\item If either there was no block for which $V.X$ saw $n - p - 2f$ votes or $V.X$ was not able to SP-lock some transaction in the block, $V.X$ forms a block to propose to Quorum Consensus based on the fallback blocks received.
\end{romanenumerate}

In particular, a transaction $tx$ is included in the Quorum Consensus proposal of $V.X$ if and only if $V.X$ saw $tx$ in at least $n - 2f$ fallback blocks and (in the UA-RA case) was able to SP-lock $tx$.

After deciding a block $B$ in the Quorum Consensus, for each UA-RA transaction $tx \in B$ from user $A$, $V.X$ sends $tx$ to $V.A$, who proposes SP-Locks and proposes it to $A$'s Transaction Agreement instance number $tx.sn$. If a conflicting transaction $tx'$ was SP-Locked before, then $tx'$ is proposed.

Upon deciding in the Transaction Agreement instance number $tx.sn$, $V.A$ notifies $V.X$ whether $tx$ was decided or not. After receiving all such notifications, $V.X$ decides on the block $B_{POA}$, which is formed from $B$ but omitting those UA-RA transactions which were not decided in their corresponding Transaction Agreement instances.

\subparagraph*{Intuition.} The reason to broadcast a fallback block is for verifiers to know the pools of each other, which would allow for a ``good'' proposal to Quorum Consensus in case the fast path fails. More precisely, knowing each other's pools, verifiers will propose ``popular'' transactions to Quorum Consensus, making them committed in the slow path and thus ensuring liveness.

The $3\Delta$ threshold consists of $\Delta$ for the leader's broadcast, $\Delta$ for verifiers' votes, and $\Delta$ for the possible shift in times at which the leader and verifier initiate the primitive. 

One needs a Transaction Agreement after Quorum Consensus to preclude committing conflicting UA-RA transactions on different reactive actors. Note that if the fast path succeeds, it ensures that no conflicting transaction can be committed; hence, we only need the Transaction Agreement in case the fast path fails.  

We provide a high-level pseudocode for POA for a reactive actor $X$ in  \cref{alg:poa-high-level-vx,alg:poa-high-level-va}. For a precise description, please see Appendix~\ref{sec:appendix:alg:poa}.

\section{Slow Path}
\label{sec:slow path}
In this section, we describe the two components of the slow path of POA and POB: Quorum Consensus and Transaction Agreement.
\subsection{Quorum Consensus}
\label{sec:quorum-consensus}

\begin{interface}
    - \texttt{function} $\textsc{Propose}(k, B)$: start the $k$-th instance with proposal $B$\\
    - \texttt{callback} $\textsc{Decide}(k, B)$: decide a block $B$ in $k$-th instance
\end{interface}

A Quorum Consensus is a partial-synchrony consensus algorithm, meaning it exposes an interface of proposing and deciding a block, and it satisfies Agreement, Termination, and Quorum Validity:

\begin{property}[Quorum Validity]
\label{prop:quorum validity}
    (I) If a transaction $tx$ is such that $tx$ is in the proposal of every honest validator, then $tx$ is included in the decided block.

    (II) If some transaction $tx$ is included in the decided block, then $tx$ is in the proposal of at least $n - 3f$ honest validators.
\end{property}

\begin{claim}
    There exists an algorithm that solves consensus with Quorum Validity.
\end{claim}
\begin{proof}
    The proof applies Theorem~5 and Definition~2 of~\cite{validity_of_consensus} to Quorum Validity. Throughout the proof we use terms and notation from~\cite{validity_of_consensus}. 

    Consider some input configuration $c \in \mathcal{I}_{n - f}$. We claim that a block $B$ consisting of transactions that are present in at least $n - 2f$ proposals in $\mathcal{I}_{n - f}$ belongs to $\bigcap_{c' \in sim(c)}val(c')$. Consider an input configuration $c'$ with at least $n - f$ elements (if there are less than $n - f$ elements in $c'$, then any block is admissible). We deduce that $|\pi(c) \cap \pi(c')| \geq n - 2f$, therefore, for each $tx \in B$, $tx$ is present in at least $n - 3f$ proposals of $c'$, therefore, can be included in the block decided for $c'$. Conversely, if some transaction is present in every proposal of $c'$, then it is present in at least $n - 2f$ proposal of $c$ and hence included in $B$. Thus, $B \in val(c')$.
\end{proof}

In \cref{sec:high-throughput-quorum-consensus} we discuss how modern high-throughput BFT protocols could be adapted to instantitate practical high-throughput Quorum Consensus.

\subsection{Transaction Agreement}
\label{sec:transaction agreement}

Transaction Agreement is a simple consensus primitive used to provide agreement on which transaction $tx$ should be committed at sequence number $sn$ for a given user $A$. Though it is possible to agree on this proactively, i.e., prior to submitting a transaction to POA, that would imply an additional latency for a UA-RA transaction under optimistic conditions, hence, \sysname does it \emph{retroactively} instead, i.e., after the decision in the Quorum Consensus has been made. 

\begin{interface}
    - \texttt{function} $\textsc{Propose}(tx)$: propose transaction $tx$ \\ 
    - \texttt{callback} $\textsc{Decide}(tx)$: decide a transaction $tx$ 
\end{interface}

The Transaction Agreement primitive should satisfy the specification of a partially synchronous multi-valued Strong Byzantine Agreement, namely:
\begin{itemize}
    \item Every process can propose and decide a (not necessarily binary) value.
    \item \textit{Agreement.} If two honest processes decide $v$ and $v'$ respectively, then $v = v'$.
    \item \textit{Quorum Validity.} If all honest processes propose the same value, only that value can be decided. Furthermore, if the value is decided, it was proposed by at least $n - 3f$ honest validators.
    \item \textit{Termination.} Given a partially synchronous network, every correct process eventually decides. 
\end{itemize}

This agreement primitive can be implemented as a special case of Quorum Consensus with blocks of size one.

\section{Transaction Processing}
\label{sec:parallel-tx-processing}


\subsection{Transaction Execution Properties}
\label{sec:tx-execution-properties}

First, we informally present several properties of transaction processing in \sysname.
These are grouped into \emph{correctness} properties, which are common among many distributed systems,
\emph{Fast Execution} properties, demonstrating the ability to process transactions with low latency,
and \emph{Parallelization} properties, showing the ability to make progress independently at different actors.
Full formal definitions as well as proofs of all properties are given in \cref{sec:proofs-transaction-execution}.

\subparagraph*{Correctness Properties.}
\label{sec:correctness-properties}

Reactive actors adhere to the Agreement, Validity, Total Order, and Integrity properties of Total Order Broadcast~\cite{defago2004total}, with respect to emission and execution of transactions.
User actors instead follow the Agreement, Validity, and Integrity properties of Byzantine Reliable Broadcast~\cite{brb}.

\subparagraph*{Fast Execution Properties.}
\label{sec:fast-execution-properties}

\sysname is designed so that under optimistic conditions, transactions are executed with low latency.
\cref{tbl:fast-execution} summarizes conditions needed for a transaction $tx$ of a given type to be executed fast. 
Those conditions are:
\emph{Honest author} (in case of user transactions) --- the user who issued a transaction is honest,
\emph{synchrony} --- the system is after GST,
\emph{honest leader} (in case of UA-RA or RA-RA with recipient $X$) --- there exists an honest validator $L$ who will start a POA instance for $X$ as a leader as soon as it has $tx$ in its pool,
\emph{good pool} --- for every UA-RA transaction $tx' \in L.X.pool$ a user who emitted $tx'$ is honest.
The \emph{latency} column shows how many communication steps are needed before every honest validator executes $tx$.
The \emph{resilience} column shows how many misbehaving validators the system can tolerate. 

\begin{table}[!ht]
    \centering
    \renewcommand{\arraystretch}{1.} 
    \setlength{\tabcolsep}{7pt} 

    \begin{tabular}{lccccccc}
        \toprule
        {\textbf{\begin{tabular}[c]{@{}l@{}}Transaction\\Type\end{tabular}}} 
        & {\textbf{\begin{tabular}[c]{@{}l@{}}Honest\\Author\end{tabular}}} 
        & \textbf{Synchrony} 
        & {\textbf{\begin{tabular}[c]{@{}l@{}}Honest\\Leader\end{tabular}}} 
        & {\textbf{\begin{tabular}[c]{@{}l@{}}Good\\Pool\end{tabular}}} 
        & \textbf{Latency}
        & \textbf{Resilience} \\
        \midrule
        UA    & \checkmark & \checkmark &   &   & $2 \delta$ & $\leq p$ \\
        RA-RA &   & \checkmark & \checkmark & \checkmark & $2 \delta $  & $\leq p$ \\
        UA-RA & \checkmark & \checkmark & \checkmark & \checkmark & \begin{tabular}[c]{@{}c@{}}$2\delta^*$ / $3\delta^\dagger$ \end{tabular} & $\leq p$ \\
        \bottomrule
    \end{tabular}
    \caption{Conditions and execution latency for the fast-path of different transaction types. 
    Latency across validators is denoted $\delta$, while latency within a validator is ignored.
    $^*$ When emitted by the leader. 
    $^\dagger$ For other validators. }
    \label{tbl:fast-execution}
\end{table}
There are instances where the \emph{good pool} condition is not met due to user misbehavior, preventing fast transaction execution.
Specifically, this occurs only when honest validators see conflicting transactions (Line \ref{line:poa-fp-lock-check}, Algorithm \ref{alg:poa}).
In such cases, an honest validator will have evidence of the user's misbehavior
(namely, two signed conflicting transactions) and may initiate punishment (e.g. destroy the gas object). We emphasize that users can only affect the liveness of the fast path and in this case, all honest transactions are still committed in the slow path (that is, Validity I of POA holds no matter the user's misbehavior).

\subparagraph*{Parallelization Property.}
\label{sec:parallelization-properties}

Classical blockchain systems address the double-spending problem \cite{nakamoto2008bitcoin} by totally ordering transactions. While effective, this introduces significant redundancy because many transactions are non-conflicting, meaning they can be executed in any order without affecting the outcome and therefore do not require total ordering.

Mangrove achieves optimal parallelization by ordering only the transactions that require it. It avoids ordering UA transactions entirely since these only create objects at user actors, and such operations are commutative. For reactive actor transactions, Mangrove maintains a partial order by keeping a separate ordered chain for each reactive actor. For example, given two reactive actors $X$ and $Y$ and a set of transactions $\{tx_1^X, \ldots, tx_k^X, tx_1^Y, \ldots, tx_l^Y\}$ (where $tx_i^A$ is a transaction with receiver $A$), Mangrove maintains two ordered sets $\{tx_1^X, \ldots, tx_k^X\}$ and $\{tx_1^Y, \ldots, tx_l^Y\}$ but does not impose an order between these sets.

For a system with a transaction set $T$ and for a reactive actor $X$, denote $T_X \subseteq T$ the subset of transactions with $X$ as a receiver.  The longest ordered chain in Mangrove is then $\max\limits_{X \in RAs}|T_X|$, whereas in classical systems like Bitcoin and Ethereum, the chain length is $|T|$. This length, $\max\limits_{X \in RAs}|T_X|$, is optimal unless additional assumptions are made about reactive actor behavior.


The creation of multiple short chains is beneficial in two ways. First, each chain can be maintained by a separate machine (running V.X) at each validator. This allows throughput to be scaled horizontally. 
Secondly, the throughput of an application (corresponding to an RA actor) depends solely on the length of its associated chain, and doesn't degrade when other applications experience high load.

\subsection{User Transactions}
\label{sec:user-actor-agreement}

Each user keeps a sequence number of the last issued transaction for each user actor it controls.
When issuing a new transaction $tx$ from a user actor $A$, a user must first check that $A$ has all owned objects consumed by $tx$. If a user fails to do so, $tx$ may never be executed, precluding all consecutive transactions from $A$. Note, though, that this only halts $A$ and not the rest of the system. After a user checks owned objects for $tx$, it assigns a new sequence number to it, signs it, and broadcasts it to all $V.A$-s using Parallel Optimistic Broadcast. In Parallel Optimistic Broadcast, a user sends $tx$ to all $V.A$-s, those attempt to FP-lock it, and if the FP-lock is successful, broadcast a vote for $tx$. Once a validator obtains $n - p$ votes $tx$, it fast-path decides $tx$, broadcasts a proof and stops sending messages. When unable to decide in the fast-path, validators invoke Transaction Agreement to ensure safety and liveness in case of asynchrony and validators' misbehavior. For pseudocode, see Appendix \ref{sec:pob}.


%
%
%

\begin{algorithm}[ht]
\caption{UA Transaction Processing}
\label{alg:ua transaction}
\begin{algorithmic}[1]
    \Uses\ Parallel Optimistic Broadcast $pob$
    \EndUses
    \Upon{event}{pob[A, k].Decide}{sender, tx} \label{line:pob decide} \Comment{On $V.A$}
    \If{$k = tx.sn$ \textbf{and} $\textsf{VerifySig}(A, tx)$}
        \State $pending \gets pending \cup \{tx\}$ \label{line:ua tx to pendig} \label{line:ua tx put in pending}
    \EndIf
    \EndUpon
    \UponExists{$tx \in pending : tx \text{ is UA} \textbf{ and } executed[tx.sn - 1]$ \textbf{and} $tx.consumedObjects \subseteq ownedObjects$} \label{line:ua unpending} \Comment{On $V.A$}
    \State $pending \gets pending \setminus \{tx\}$
    \State $\textsf{effects} \gets \textsf{VM.Execute}(tx.Code)$ \label{line:ua tx execute}
    \State $ownedObjects \gets ownedObjects \setminus tx.consumedObjects \cup \textsf{effects}.createdObjects$
    \State $executed[tx.sn] \gets true$
    \EndUpon
\end{algorithmic}
\end{algorithm}

For a UA transaction $tx$ from $A$ with a sequence number $sn$ and consumed objects $O_1, \ldots, O_k$, it is \emph{executed} if (i) $tx$ is decided in Parallel Optimistic Broadcast, (ii) a transaction from $A$ with a sequence number $sn - 1$ is executed and (iii) $A$ owns $O_1, \ldots, O_k$. See the pseudocode for UA transactions in Algorithm \ref{alg:ua transaction}. 

For a UA-RA transaction $tx$ from $A$ to $X$ with a sequence number $sn$ and consumed objects $O_1, \ldots, O_k$, $V.A$ \emph{sends} it to $V.X$ if $tx$ is correctly signed by $A$, $V.A$ have not seen any other transactions with sequence number $sn$, and (ii) and (iii) hold.
$tx$ is \emph{executed} as soon as it is decided in the POA of $X$. Note that execution crucially does not rely on (ii) and (iii) to hold. 
We would like to highlight here that UA-RA transactions \emph{do not use} an agreement mechanism on the sending user actor. 
See the pseudocode for UA-RA transactions in Algorithm \ref{alg:ua-ra transaction}.

\begin{algorithm}[ht]
\caption{UA-RA Transaction Processing}
\label{alg:ua-ra transaction}
\begin{algorithmic}[1]
    \Uses\ Outer Links $ol$, Inner Links $il$, POA $poa$
    \EndUses

    
    \Upon{event}{Emit UA-RA Transaction}{A, X, [O_1, \ldots, O_k], Code, Call} \Comment{On $U$}
    \If{\textbf{not} $executed[A][sns[A] - 1]$ \textbf{or} $\{O_1, \ldots, O_k\} \not\subseteq ownedObjects[A]$} \label{line:user ra transaction checks}
        \State \Return Error(``Not possible to emit'')
    \EndIf
    \State $tx \gets \mathsf{Sign}(\langle A, sns[A], X, [O_1, \ldots, O_k], Code, Call\rangle)$
    \State $sns[A] \gets sns[A] + 1$
    \ForAll {$V \in \mathcal{V}$}
        \textbf{trigger} $\langle \mathrm{ol.Send} \mid V, tx\rangle$
    \EndFor
    \EndUpon

    
    \Upon{event}{ol.Deliver}{tx: \text{UA-RA Transaction}} \Comment{On $V.A$}
    \If { \textbf{not} VerifySign($A$, $tx$) \textbf{or}
        $FPLocked[tx.sn] \neq \bot$} \label{line:ua-ra checks}
        \Return
    \EndIf
    \State $FPLocked[tx.sn] \gets tx$ \label{line:ua ra tx put tx in transactions}
    %
    \State \textbf{await} $executed[tx.sn-1]$ \textbf{and} $tx.consumedObjects \subseteq ownedObjects$ \label{line:ua-ra transaction unpending} 
    \State \textbf{trigger} $\langle\mathrm{il.Send} \mid tx.receiver, tx \rangle$ \label{line:ua-ra transaction send}
    \EndUpon

    
    \Upon{event}{il.Deliver}{tx, A} \Comment{On $V.X$}
    \If{$tx \notin executed$}
        $pool \gets pool \cup \{tx\}$\label{line:ua-ra transaction pool}
    \EndIf
    \EndUpon

    
\end{algorithmic}
\end{algorithm}


\subsection{Reactive Actor Transactions}
\label{sec:ra-transactions}

When a reactive actor $V.X$ decides a block, it starts executing transactions of that block in a deterministic order.
Some transactions might create outgoing transactions when executed.
Upon creating an outgoing transaction $tx$ that consumes owned objects $O_1, \ldots, O_k$, $V.X$ checks its owned objects.
If it owns all required objects, $V.X$ sends $tx$ to the receiver via inner links, and receiver adds $tx$ to the pool.
If some objects that transaction consumes are missing, then this transaction is immediately dropped and has no effect. For the pseudocode, please see Algorithm \ref{alg:ra-tx-execution} in \cref{sec:execution}.


\section{Discussion and Outlook}
\label{sec:discussion}

In general, a transaction may result in a \emph{cascade} of multiple subsequent transactions.
Classical systems guarantee that such cascades appear to be executed atomically and in isolation.
Moreover, they immediately execute the whole cascade of a transaction after agreeing on the initial transaction,
whereas \sysname (in the worst case) performs a separate agreement for each individual transaction in the cascade.
Therefore, in the case of deep cascades, we expect classical solutions to complete the execution of the cascade faster.
Empirical study of the workload types under which either approach performs better is subject to future research.

Finally, we acknowledge that the bit complexity and message complexity of POA are relatively high. However, the primary objective of this work is to demonstrate the feasibility of a system that supports smart contracts while enabling maximal parallelization. Optimizing these complexities is an important consideration, which we leave as future work.

\bibliography{bibliography}

\appendix

\section{Model Expressiveness}
\label{sec:expressiveness}

Most common applications can be easily adapted from something resembling Ethereum's smart contract model
to our model with a gain in scalability and without a loss in expressiveness apart from the obvious loss of free atomic composability.
For the most part, this is equivalent to how these applications would be designed for the Move VM~\cite{move}, on the Sui~\cite{sui_website} and Aptos~\cite{aptos_website} blockchains.

\begin{itemize}
    \item \textbf{Token:} A token is primarily defined by a data type for owned objects that adheres to a specific interface.
        Additionally, a token might have an associated reactive actor responsible for minting new tokens or handling a reserve.
    \item \textbf{Decentralized Exchange:}
        Each liquidity pool of a decentralized exchange should be its own reactive actor.
        This allows asynchronous access to different liquidity pools.
        One popular heavily congested contract does not impede access to any of the other liquidity pools.
    \item \textbf{NFT Marketplace:}
        The marketplace would be a reactive actor owning all the (NFT) objects currently for sale.
        Buyers can interact by sending accepted tokens as payment and receiving ownership of the desired object.
        Sellers can interact by sending it new  objects to put up for sale or delisting items, receiving back the object.
\end{itemize}

More generally, we argue the replicated actor model (\cref{sec:model}) does not lose expressiveness.
This can be shown by providing a general framework of translating applications from Ethereum's smart contract model to our model.
However, atomic composability is not inherently guaranteed at the protocol level.
If it is intended, it has to be specifically designed into applications and users could be charged additional fees at the application layer for the privilege.

\subparagraph*{Gas.}

To incentivize validators to perform computations, our system uses \emph{gas objects}, a type of owned object.
Each transaction needs to consume at least a gas object.
As validators process the $Call$ and $Code$ fields, they are compensated through fees deducted from the provided gas object.
The amount of gas required depends on the complexity of the transaction,
ensuring fair compensation for resource-intensive tasks.
Design and study of specific game-theoretic mechanisms in our system model is outside the scope of this work.

\subparagraph*{Atomic Composability.}

Full atomic composability could be trivially achieved by keeping all composable state in a single reactive actor.
However, this is not realistic when considering an open, diverse and growing ecosystem of applications.
More importantly, it does not make use of the scalability advantages of our model.
We can instead give users the option to request composability across reactive actors on-demand.

\subparagraph*{Locking.}

To this end we define a \emph{locking} pattern that applications may independently opt into.
Locking is unique in that it allows idle blocking while waiting for other calls to return.
A reactive actor that supports locking has associated state
$lockHolder$, $lockCollateral$, $lockStartTime$, $lockPrice$,
and exposes functions $lock()$, $unlock()$, $getLockHolder()$, and $getLockPrice()$.
The locking fee is the product of the current price for locking and the time the lock is held.
The lock is valid until the required fee has exceeded the posted collateral.
The price for locking can be updated by the reactive actor and may depend on the application as well as current usage statistics.
For example, a decentralized exchange may count accesses or trading volume to a given liquidity pool and set the fee proportional to its recent popularity.

Any $lockPrice$ larger than 0 prevents a permanent deadlock on the reactive actor.
However, the application protocol designers need to ensure that at any time the price is fair, to prevent abuse of this feature.
The price for holding a lock could even increase super-linearly in the time it is held.
This would further discourage continuously preventing others from acquiring the lock.

This would require some support on the protocol side as well.
At least, each transaction should receive a timestamp agreed upon by the validators.
This would serve as the basis to determine how long each lock is held.
Also, execution of other transactions at that reactive actor should be delayed until after the lock is released.
Maybe with the exception of a call checking whether a lock is currently being held.

\subparagraph*{Hard Example: Flash Loan.}

Some particularities of the strong atomic transaction model are not easily translatable.
Locking is necessary but not sufficient to enable full atomic composability.
For example, there is no way to implement a flash loan using just the above locking interface.

A flash loan is a loan where there is no risk of default because the loan is only valid within an atomic and isolated transaction.
In this case, the lender needs to be sure that the entire transaction may only commit if the loan is paid back.
If not, it should be reverted and the loan paid back.

\subparagraph*{Fully Atomic Transaction Cascades.}

Locking can be extended to provide atomic and isolated execution of entire transaction cascades.
To be able to revert transactions, an atomic transaction execution context is needed.
Most importantly, objects stay within the execution context until the transaction commits.
Otherwise, other transactions might be able to see partial effects of the cascade before it commits.
This would also make reversion on abort impossible without affecting other transactions.

Within an atomic transaction context all additional calls that are made are considered to be atomic and part of the same execution context.
These calls need to hold a lock of the recipient reactive actor.
The context is passed along the entire transaction cascade.
Any transaction within an atomic transaction context may only spawn new transactions within the same context.
Locks can only be freed at the end of the transaction (i.e. once it commits or aborts).
Conversely, if any of the locks run out of collateral, the transaction aborts.
If a particular transaction aborts, the entire transaction cascade aborts.
Importantly, all consumed objects are returned to the sender.

Using this extension of our model even flash loans can be realized.
The only additional thing that is required for safety is that the \texttt{lend()} method of the flash loan reactive actor can indicate that it needs to be called from an atomic execution context and may revert.

\section{High-Throughput Quorum Consensus}
\label{sec:high-throughput-quorum-consensus}

A reactive actor that is at the core of a popular application may experience high contention over a long time period.
Continuously trying to get transactions accepted on the fast-path in this case inhibits both throughput and latency.
On the other hand, continuously running a chained high-throughput consensus protocol at each reactive actor, where most instances only produce empty blocks, is a significant and unnecessary burden.
To alleviate this, reactive actors should be able to individually toggle between single-shot POA and a high-throughput consensus mechanism.
This could happen automatically based on certain heuristics (configured either at the system level or by each RA).
Switching from uncontested to contested mode can be done by deciding a special message in the same manner as any transaction.
Switching from contested back to uncontested mode can be done via the high-throughput consensus' epoch closing mechanism.

Existing high-throughput consensus protocols, like Narwhal \& Tusk~\cite{narwhal}, Bullshark~\cite{bullshark}, Shoal/Shoal++~\cite{shoal,shoal++} and Mysticeti~\cite{mysticeti} 
can be extended to achieve part (II) of Quorum Validity, which is what differentiates it from regular Validity.
For this, a simple consensus rule can be added.
This rule restricts whether an ordered transaction is actually delivered.
In addition to the voting on blocks necessary for ordering, validators directly vote on transactions.
Honest validators vote for a transaction if and only if there are no conflicting transactions in their pool.
Only transactions reaching a threshold of direct votes are delivered for execution.
If this threshold ensures a quorum intersection of at least $f+1$, no conflicting transactions can be committed.
This is the same rule that is added in \textsc{Mysticeti-FPC}, to make the general consensus compatible with the reliable broadcast fast-path allowed for owned-object-only transactions.

\section{Parallel Optimistic Broadcast}
\label{sec:pob}

The Parallel Optimistic Broadcast (POB) primitive has an interface consisting of:

\begin{interface}
    - \texttt{function} $\textsc{Broadcast}(sn, tx)$: broadcast $sn$-th transaction \\
    - \texttt{callback} $\textsc{Decide}(sn, tx)$: decide transaction $tx$
\end{interface} 

\begin{figure}[ht]
    \centering
    \includegraphics[width=0.7\linewidth]{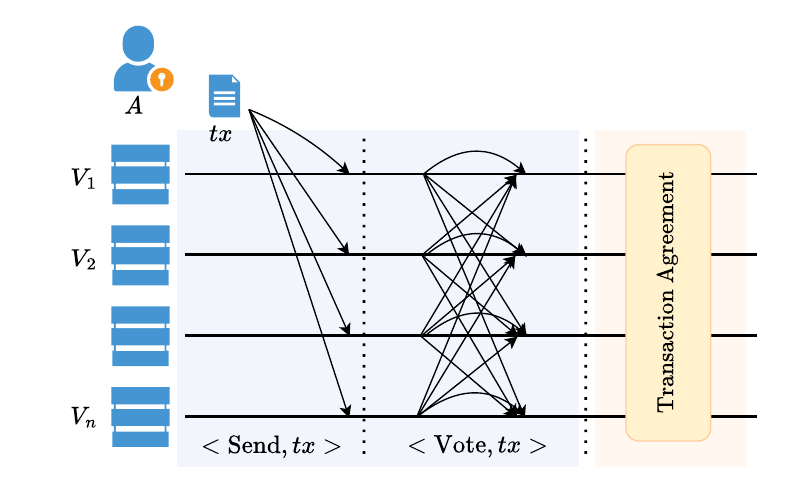}
    \caption{POB scheme. In the fast path (blue rectangle), a user $A$ broadcasts their transaction $tx$.
        Then, validators cast their vote. In case the fast path fails, validators participate in the slow path (orange box), which consists of one instance of Transaction Agreement.}
    \label{fig:pob_scheme}
\end{figure}

To issue a UA transaction $tx$, user $A$ broadcasts it among all $V.A$-s. Subsequently, $V.A$-s FP-lock $tx$ and vote for it, and if some $V.A$ obtains $n - p$ votes, they can decide $tx$ and stop sending messages. After receiving $n - p - 2f$ votes for $tx$, $V.A$ SP-locks $tx$ and proposes it to a Transaction Agreement instance number $tx.sn$ of $A$.

\begin{algorithm}[!ht]
\caption{Parallel Optimistic Broadcast} \label{alg:ua transaction broadacast}
\begin{algorithmic}[1]
    \Uses\ Perfect Links \textit{ol}, Transaction Agreement \textit{ta}
    \EndUses
    \Statex

    \Upon{event}{broadcast}{tx} \Comment{On User $A$}
        \ForAll{$V \in \mathcal{V}$}
            \State \textbf{trigger} $\langle ol.Send \mid V.A,  tx \rangle$
        \EndFor
    \EndUpon

    \Statex


    \Upon{event}{ol.Deliver}{A,  tx} \Comment{On $V.A$}
        \If{$FPLocked[tx.sn] \neq \bot$ \textbf{and} $FPLocked[tx.sn] \neq tx$} \label{line:pob check tx seen} \label{line:pob check fp lock}
            \State $vote \gets \mathsf{Sign}(\bot)$
        \Else
            \State $FPLocked[tx.sn] \gets tx$ \label{line:pob fix tx}
            \State $vote \gets \mathsf{Sign}(\mathsf{Vote}(tx))$
        \EndIf
        \ForAll{$V \in \mathcal{V}$}
            \State \textbf{trigger} $\langle ol.Send \mid V.A, vote \rangle$ \label{line:pob broadcast vote}
        \EndFor
    \EndUpon

    \Statex

    \UponTrue{received $n - p$ votes for $tx$} \Comment{On $V.A$}
        \If{$pobDecided$} \label{line:pob decided check 1}
            \State \Return
        \EndIf
        \State $proof(tx) \gets aggregate(n - p\ votes)$
        \ForAll{$V \in \mathcal{V}$}
            \State \textbf{trigger} $\langle ol.Send \mid V.A, proof(tx) \rangle$
        \EndFor
        \State \textbf{trigger} $Decide(tx)$ \label{line:pob fp decide}
        \State $pobDecided \gets True$
    \EndUpon

    \Statex

    \Upon{event}{ol.Deliver}{V', proof(tx)} \Comment{On $V.A$}
        \If{$verify(proof(tx))$}
            \ForAll{$V'' \in \mathcal{V}$}
                \State \textbf{trigger} $\langle ol.Send \mid V''.A, proof(tx) \rangle$
            \EndFor
            \State \textbf{trigger} $Decide(tx)$ \label{line:pob decide by proof}
        \EndIf
    \EndUpon
    
    \Statex

    \UponTrue{received $n - p - 2f$ votes for $tx$} \Comment{On $V.A$}
        \If{$SPLocked[tx.sn] \neq \bot$ \textbf{and} $SPLocked[tx.sn] \neq tx$}
            \State \Return
        \EndIf
        \State $SPLocked[tx.sn] \gets tx$ \label{line:pob sp lock}
        \State \textbf{trigger} $ta[tx.sn].Propose(tx)$ \label{line:pob propose to ua}
    \EndUpon

    \Statex

    \Upon{event}{ta[tx.sn].Decide}{tx} \Comment{On $V.A$}
        \If{$pobDecided$} \label{line:pob decided check 2}
            \State \Return
        \EndIf
        \State \textbf{trigger} $Decide(tx)$ \label{line:pob sp decide}
        \State $pobDecided \gets True$
    \EndUpon
\end{algorithmic}
\end{algorithm}

Parallel Optimistic Broadcast adheres to the following properties.

\begin{property}[Agreement]
    If two honest validators decide $tx_1$ and $tx_2$ in the same instance of Parallel Optimistic Broadcast, then $tx_1 = tx_2$.
\end{property}
\begin{proof}
    If both validators decide in the Transaction Agreement, the Agreement follows from the Agreement of Transaction Agreement. 

    If both validators decide in the fast path (here, we also call a decision done after receiving a proof on Line~\ref{line:pob decide by proof} a fast-path decision), that means that each of them obtained $n - p$ votes for $tx$ (Line~\ref{line:pob fp decide}), meaning there is at least one common honest vote, and hence $tx_1 = tx_2$ since no honest validator issues conflicting votes (Line~\ref{line:pob check fp lock}).

    If $tx_1$ was decided in the fast path, it means it got $n - p$ votes, so no conflicting transaction can obtain $n - p - 2f$ votes (the intersection is $(n - p) + (n - p - 2f) - n = n - 2p - 2f \geq f + 1$), hence no honest verifier will propose it to Transaction Agreement, hence, by the Validity property of the Transaction Agreement, $tx_2$ can't be decided. 
\end{proof}

\begin{property}[Integrity]
    A transaction $tx$ is decided at most once in the given instance and only if it was broadcast by the user.
\end{property}
\begin{proof}
    The only once part is ensured through checks (Lines~\ref{line:pob decided check 1} and \ref{line:pob decided check 2}), and if the user didn't emit the transaction, it will receive at most $f$ votes, which is not sufficient neither for the fast path ($n - p > f$), nor to propose it to the Transaction Agreement ($n - p - 2f > f$).
\end{proof}

\begin{property}[Validity]
    If an honest user broadcasts a transaction $tx$, it is eventually decided.
\end{property}
\begin{proof}
    Since an honest user does not issue conflicting transactions, all honest verifiers will eventually FP-lock $tx$ (Line~\ref{line:pob check fp lock}) and broadcast their vote. Now, if at any time an honest verifier gets $n - p$ votes, it broadcasts the proof, hence all honest verifiers will eventually decide. If none of the honest verifiers receive $n - p$ votes, since the user is honest, they all will eventually get $n - f \geq n - p - 2f$ votes for $tx$ and hence will propose it to the Transaction Agreement. Therefore, by the Validity property of Transaction Agreement, $tx$ will be eventually decided in it and hence in Parallel Optimistic Broadcast (Line~\ref{line:pob sp decide}).
\end{proof}

\begin{property}[Fast Termination]
    Given the system is after GST, a user issuing the transaction is honest and at most $p$ validators misbehave, every validator decides and stops sending messages in $2\delta$ time after a user broadcasts its transaction.
\end{property}
\begin{proof}
    Assume an honest user broadcasts a transaction $tx$ at time $t$. Then, since the system is after GST, all honest verifiers will receive it, and, since an honest user can not issue conflicting transactions, will FP-lock it (Line~\ref{line:pob fix tx}) and broadcast their vote for it (Line~\ref{line:pob broadcast vote}). Therefore, by the time $t + 2\delta$, every honest verifier will receive at least $n - p$ votes for $tx$ and will fast-path decide it (Line~\ref{line:pob fp decide}). 
\end{proof}

\section{Parallel Optimistic Agreement (Extended)}
\label{sec:appendix:alg:poa}

\begin{algorithm}[!ht]
\caption{Parallel Optimistic Agreement (Part 1)}
\label{alg:poa}
\begin{algorithmic}[1]
    \Uses\ Outer Links $ol$, Inner Links $il$, Quorum Consensus $qc$, Transaction Agreement $ta$
    \EndUses

    \Function{Initiate}{k, pool} \Comment{On $V.X$}
        \State $poaInstance \gets k$
        \If{$leaderOracle(k) = V$}
            \ForAll{$V' \in \mathcal{V}$}
                $ol.Send(V'.X, \langle k, B \coloneqq pool\rangle)$ \label{line: poa leader broadcast}
            \EndFor
        \EndIf
        \State $Timer.Restart()$
    \EndFunction

    \Statex
    
    \Upon{event}{$ol.Deliver$}{L, \langle k, B\rangle} \Comment{On $V.X$}
        \If{$leaderOracle(k) = L$}
            $block[k] \gets B$
        \EndIf
    \EndUpon

    \Statex
    
    \UponTrue{$poaInstance = k$ \textbf{and} $block[k] \neq \bot$ \textbf{and} $Timer \leq 3\Delta$ \textbf{and} $\forall tx \text{ RA-RA transaction} \in block[k]:\ tx \in pool$} \label{line:poa receive leader's block condition} \Comment{On $V.X$}
        \State $repliesFPLock[k] \gets 0$
        \State $successFPLock[k] \gets true$
        \ForAll{$tx: \text{UA-RA transaction} \in block[k]$} \Comment{In parallel}
            \State \textbf{trigger} $\langle il.Send \mid FPLock(tx), tx.sender\rangle$ \label{line:poa try fp-lock}
        \EndFor    
    \EndUpon

    \Statex
    
    \Upon{event}{$il.Deliver$}{FPLock(tx), X} \label{line:poa-fp-lock-check} \Comment{On $V.A$}
        \If{\textbf{not} $\mathsf{VerifySig}(tx.sender, tx)$}  \label{line:poa check user sign}
            \State \textbf{trigger} $\langle il.Send \mid \textbf{false}, X \rangle$; \Return
        \EndIf
        \If{$FPLocked[tx.sn] = \bot$}
            $FPLocked[tx.sn] \gets tx$
        \EndIf
        \If{$FPLocked[tx.sn] \neq tx$}
            \State \textbf{trigger} $\langle il.Send \mid \textbf{false}, X \rangle$; \Return
        \EndIf
        \If{within $\Delta$ time $executed[tx.sn - 1] \textbf{ and } tx.objects \subseteq ownedObjects$} \label{line:poa approve fp lock condtion}
            \State \textbf{trigger} $\langle il.Send \mid \textbf{true}, X \rangle$
        \Else
            \State \textbf{trigger} $\langle il.Send \mid \textbf{false}, X \rangle$
        \EndIf
    \EndUpon

    \Statex

    \Upon{event}{$il.Deliver$}{status, A} \Comment{On $V.X$}
        \State $repliesFPLock[k] \gets repliesFPLock[k] + 1$
        \State $successFPLock[k] \gets successFPLock[k] \wedge status$
    \EndUpon

    \Statex

    \UponTrue{$repliesFPLock[k] == |\{tx: \text{UA-RA transaction} \in block[k]\}|$} \label{line:poa trigger vote}\Comment{On $V.X$}
        \State $Vote \gets block[k]$
        \If{\textbf{not} $successFPLock[k]$}
            $Vote \gets \bot$
        \EndIf
        \State $B_\mathsf{fb} \gets pool$ \label{line:poa let fallback block}
        \ForAll{$V' \in \mathcal{V}$}
            \State \textbf{trigger} $\langle ol.Send \mid [ k, \mathsf{Vote}, B_{fb}], V'.X \rangle$ \label{line:poa broadcast vote}
        \EndFor
    \EndUpon
\algstore{poa}
\end{algorithmic}
\end{algorithm}
\addtocounter{algorithm}{-1}
\begin{algorithm}[!ht]
\caption{Parallel Optimistic Agreement (Part 2)}
\begin{algorithmic}[1]
\algrestore{poa}
    \Upon{event}{$ol.Deliver$}{U, \langle k, \mathsf{Vote}, B_\mathsf{fb} \rangle} \Comment{On $V.X$}
        \State $votes[k] \gets votes[k] \cup \{\mathsf{Vote}\}$
        \If{$\mathsf{Vote} \neq \bot \land |\{v \in votes[k]\ |\ v = \mathsf{Vote}\}| = n - p$} \label{line:poa-fast-path-condition}
            \State $proof(B) \gets aggregate(n - p\ votes)$
            \ForAll{$V'' \in \mathcal{V}$}
                \State \textbf{trigger} $\langle ol.Send \mid V''.X, proof(B) \rangle$ \label{line:poa broadcast proof}
            \EndFor
            \State \textbf{trigger} $\langle poa.Decide \mid k, B\rangle$ \label{line:poa-decide-optimistic}
        \EndIf
        \State $fallbackBlocks[k] \gets fallbackBlocks[k] \cup \{B_\mathsf{fb}\}$
    \EndUpon

    \Statex

    \Upon{event}{ol.Deliver}{V', proof(tx)} \Comment{On $V.X$}
        \If{$verify(proof(tx))$}
            \ForAll{$V'' \in \mathcal{V}$}
                \State \textbf{trigger} $\langle ol.Send \mid V''.A, proof(tx) \rangle$
            \EndFor
            \State \textbf{trigger} $\langle poa.Decide \mid k, B\rangle$ \label{line:poa transaction broadcast decide by proof}
        \EndIf
    \EndUpon

    \Statex
		
    \UponTrue{$poaInstance = k$ \textbf{and} $Timer > 3\Delta$ \textbf{and} $|votes[k]| = n - f$} \label{line:poa slow path condition} \Comment{$V.X$}
        \If{$\nexists B: |\{v \in votes[k]\ |\ v = B\}| \geq n-p-2f$}
            \State $needFallbackBlocks \gets true$
            \Return
        \EndIf
        \State $candidateB[k] \gets B$ such that $|\{v \in votes[k]\ |\ v = B\}| \geq n-3f$
        \State $repliesSPLock[k] \gets 0$
        \State $successSPLock[k] \gets true$
        \ForAll{$tx: \text{UA-RA transaction} \in candidateB[k]$} \Comment{In parallel}
            \State \textbf{trigger} $\langle il.Send \mid SPLock(tx), tx.sender\rangle$ \label{line:poa sp lock for fp}
        \EndFor
    \EndUpon

    \Statex

    \Upon{event}{$il.Deliver$}{SPLock(tx), X} \Comment{On $V.A$}
        \If{$SPLocked[tx.sn] = \bot$} \label{line:poa sp lock condition}
            $SPLocked[tx.sn] \gets tx$
        \EndIf
        \State $status  \gets SPLocked[tx.sn] == tx $
        \State \textbf{trigger} $\langle il.Send \mid status, X \rangle$
    \EndUpon

    \Statex

    \Upon{event}{$il.Deliver$}{status, A} \Comment{On $V.X$}
        \State $repliesSPLock[k] \gets repliesSPLock[k] + 1$
        \State $successSPLock[k] \gets successSPLock[k] \wedge status$
    \EndUpon

    \Statex

    \UponTrue{$repliesSPLock[k] == |\{tx: \text{UA-RA transaction} \in candidateB[k]\}|$} \Comment{On $V.X$}
        \If{$successSPLock[k]$} \label{line:poa if success sp lock}
            \State \textbf{trigger} $\langle qc.Propose \mid B \rangle$ \label{line:poa propose to qc for fast path};
            \Return
        \EndIf
        \State $needFallbackBlocks \gets true$
    \EndUpon
    \algstore{poa}
\end{algorithmic}
\end{algorithm}
\addtocounter{algorithm}{-1}
\begin{algorithm}[!ht]
\caption{Parallel Optimistic Agreement (Part 3)}
\begin{algorithmic}[1]
\algrestore{poa}
    \UponTrue{$needFallbackBlocks$} \Comment{$V.X$}
        \State $candidateB_2[k] \gets \{tx \mid |\{B_\mathsf{fb} \in fallbackBlocks[k] \mid tx \in B_\mathsf{fb} \}| \geq n - 2f \}$ \label{line:poa get slow path proposal from FB blocks}
        \State $SPLockedTxs[k] \gets \emptyset$
        \State $repliesSPLock_2[k] \gets 0$
        \ForAll{$tx: \text{UA-RA transaction} \in candidateB_2[k]$}
            \State \textbf{trigger} $\langle il.Send \mid SPLock_2(tx), tx.sender\rangle$ \label{line:poa sp lock fallback}
        \EndFor
    \EndUpon

    \Statex

    \Upon{event}{$il.Deliver$}{SPLock_2(tx), X} \Comment{On $V.A$}
        \If{$SPLocked[tx.sn] = \bot$} \label{line:poa sp lock condition 2}
            $SPLocked[tx.sn] \gets tx$
        \EndIf
        \State $status  \gets SPLocked[tx.sn] == tx $
        \State \textbf{trigger} $\langle il.Send \mid \langle status, tx\rangle, X \rangle$
    \EndUpon{}

    \Statex

    \Upon{event}{$il.Deliver$}{\langle status, tx \rangle, A} \Comment{On $V.X$}
        \If{$status$}
            \State $SPLockedTxs[k] \gets SPLockedTxs[k] \cup \{tx\}$
        \EndIf
        \State $repliesSPLock_2[k] \gets repliesSPLock + 1$
    \EndUpon

    \Statex

    \UponTrue{$repliesSPLock_2[k] == |\{tx: \text{UA-RA transaction} \in candidateB_2[k]\}|$} \Comment{On $V.X$}
        \State $QCPropose \gets SPLockedTxs[k] \cup \{tx \in candidateB_2 \mid tx \text{ is RA-RA transaction}\}$ \label{line:poa add ra-ra txs to qc proposal}
        \State \textbf{trigger} $\langle qc.Propose \mid QCPropose\rangle$ \label{line:poa propose to qc fallback}
    \EndUpon

    \Statex
		

    \Upon{event}{$qc.Decide$}{B_{qc}} \Comment{On $V.X$}
        \ForAll{$tx: \text{UA-RA transaction} \in B$}
            \State $A \gets tx.sender$
            \State \textbf{trigger} $\langle il.Send \mid UAInitiate(tx), V.A \rangle$ \label{line:poa-ua-initiate}
        \EndFor
    \EndUpon

    \Statex

    \Upon{event}{$il.Deliver$}{UAInitiate(tx), V.X} \Comment{On $V.A$}
        \If{$SPLocked[tx.sn] = \bot$}
            $SPLocked[tx.sn] \gets tx$ \label{line:poa sp lock for UA}
        \EndIf
        \State \textbf{trigger} $\langle ta[tx.sn].Propose \mid SPLocked[tx.sn] \rangle$ \label{line:poa propose ua}
    \EndUpon
    
    \Statex

    \Upon{event}{$ta[sn].Decide$}{$tx'$} \Comment{On $V.A$}
        \If{$tx' = tx$} \textbf{trigger} $\langle il.Send \mid \text{ ``UA success''}, V.X \rangle$
        \Else\ \textbf{trigger} $\langle il.Send \mid \text{ ``UA fail''}, V.X \rangle$
        \EndIf
    \EndUpon
    
    \Statex
    
    \UponTrue{received all Transaction Agreement responses} \Comment{On $V.X$}
        \State $B_{POA} \gets B_{qc} \setminus \{tx \mid tx \text{ failed in Transaction Agreement}\}$ \label{line:poa-ua-final-block}
        \State \textbf{trigger} $\langle poa.Decide \mid poaInstance, B_{POA}\rangle$ \label{line:poa decide slow}
    \EndUpon
    
    \Statex

\end{algorithmic}
\end{algorithm}

This section provides the full pseudocode and proofs showing that \cref{alg:poa}, as presented in \cref{sec:poa-algorithm}, achieves the properties laid out in \cref{sec:poa-properties}.

\begin{lemma}
    \label{lem:huuy}
        Given $tx \in B$ was fast-path decided, no honest validator can SP-Lock a conflicting transaction $tx'$.
    \end{lemma}
    \begin{proof}
        A validator attempts to SP-lock a UA-RA transaction $tx'$ in four cases.
        Either because it received $n - p - 2f$ votes for a block containing $tx'$ (Line~\ref{line:poa sp lock for fp}), because it received $n - p - 2f$ votes for $tx'$ in Parallel Optimistic Broadcast (Line~\ref{line:pob sp lock}, \cref{alg:ua transaction broadacast}), because it received $n - 2f$ fallback blocks containing $tx'$ (Line~\ref{line:poa sp lock fallback}), or because it decided a block containing $tx'$ in Quorum Consensus (Line~\ref{line:poa sp lock for UA}).
        We want to show that given some honest validator Fast-Path decided a block containing $tx$, none of the above can hold.
        
        Let's start with $n - p - 2f$ votes.
        If a validator received $n - p - 2f$ votes for a block containing $tx'$, it means that at least $n - p - 3f \geq p + 1$ honest validators FP-locked $tx'$ (Line~\ref{line:poa try fp-lock}). The same argument holds for $tx'$ in Parallel Optimistic Broadcast (Line~\ref{line:pob fix tx}).
        But those wouldn't then vote for block $B$ that contains $tx$, since an FP-lock attempt in Line~\ref{line:poa try fp-lock} would fail.
        Hence, no validator can receive $n - p$ votes for $B$, thus no fast decision of $B$ is possible. A contradiction.
        
        Next, assume that a validator SP-locked $tx'$ due to receiving $n - 2f$ fallback Blocks with $tx'$.
        This implies that at least $n - 3f$ honest validators broadcasted a fallback block with $tx'$ and hence at least $n - 3f \geq 2p + 1$ honest validators have $tx'$ in their pool (Line~\ref{line:poa let fallback block}), meaning they have $FPLocked[tx.sn] = tx'$ (Line~\ref{line:ua ra tx put tx in transactions}) and therefore can not FP-lock $tx$, and wouldn't vote for a block containing $tx$.
        Thus, a Fast-Path decision of a block containing $tx$ is not possible, contradiction.
        
        Finally, assume that a validator SP-locked $tx'$ due to deciding a block containing $tx'$ in Quorum Consensus. By the part (II) of the Quorum Validity property, that means that at least $n - 3f \geq 2p + 1$ honest validators proposed $tx'$ to Quorum Consensus, and hence SP-Locked it for one of the first two reasons(Lines~\ref{line:poa sp lock condition} or \ref{line:poa sp lock condition 2}) which we've shown to be impossible given $tx$ was fast-path decided.
    \end{proof}

\begin{proof}[Agreement Property]
    If two validators decide blocks $B$ and $B'$ in the fast path, that means that each of them received at least $n - p$ votes (Line~\ref{line:poa-fast-path-condition}), hence there must be at least $(n - p) + (n - p) - n = 3f + 1$ common votes, therefore at least $2f + 1$ honest validators voted for both $B$ and $B'$ meaning $B = B'$ since an honest validator only votes once (Line~\ref{line:poa trigger vote}).

    Assume two honest validators $V$ and $V'$ decided blocks $B$ and $B'$ respectively in the slow path. By the agreement property of the Quorum Consensus, $V$ and $V'$ decided the same block $B_{qc}$ in Quorum Consensus, hence they have the same set of UA-RA transactions to send to Transaction Agreements and by the Agreement Property of the Transaction Agreement, the same subset $S$ of those was not decided. Thus, $B = B' = B_{qc} \setminus S$.

    Therefore, what is left to show is that if an honest validator $V$ decides $B$ in the fast path and another honest validator decides $B'$ in Quorum Consensus, then $B = B'$.

    If $V$ decides $B$ in the Fast Path, it means that it received $n - p$ votes for $B$.
    Hence among every $n - f$ votes there will be at least $n - p - 2f$ votes for $B$.
    We would like to show that every honest validator will propose $B$ to Quorum Consensus, and hence $B$ will be decided in Quorum Consensus.
    
    To do so, we need to show that a condition in Line~\ref{line:poa if success sp lock} will hold, namely that every UA-RA transaction $tx \in B$ will be successfully SP-locked.
    We show it by showing that no conflicting transaction $tx'$ can be SP-locked.

    So every honest validator will propose $B$ to the Quorum Consensus. Denote with $B_{qc}$ a block decided in Quorum Consensus. By the condition (I) of Quorum Validity, if $tx$ is in the proposal of every honest validator, then $tx$ is in the decided block, that is $tx \in B \rightarrow tx \in B_{qc}$. Denote with $S \subset B_{qc}$ a subset of UA-RA transactions in $B_{qc}$ that were not decided in the corresponding Transaction Agreement. We will show that $tx \in B \rightarrow tx \not\in S$. Indeed, assume $tx \in B$. Then, by Lemma \ref{lem:huuy}, all honest validators will propose it (Line~\ref{line:poa propose ua}) to the Transaction Agreement, and by the Quorum Validity of the Transaction Agreement $tx$ will be decided. This allows us to conclude that $B \subseteq B_{qc} \setminus S = B'$ 
	
    Next, note that $B_{qc} \subseteq B$.
    Indeed, by condition (II) of Quorum Validity, if $tx$ is in $B_{qc}$, then it was proposed by at least $n - 3f \geq 1$ honest validators, and hence, as we've showed that every honest validator proposes $B$,  $tx$ must be in $B$. This gives us $B' = B_{qc} \setminus S \subseteq B_{qc} \subseteq B$
\end{proof}

\begin{proof}[Termination Property]
    There will either be an honest validator that received $n - p$ votes for some block $B$, or there will be not. In case there will be, it will broadcast the proof, hence every honest validator will eventually receive it and will eventually decide and stop sending messages. If none of the honest validators ever receive $n - p$ votes, we argue by the termination of a slow path:

    By $3\Delta$ time after the start of the instance (Line~\ref{line:poa slow path condition}) an honest validator will propose to the Quorum Consensus (Line~\ref{line:poa propose to qc for fast path} or \ref{line:poa propose to qc fallback}).
    They eventually decide in the Quorum Consensus by the Termination property of Quorum Consensus and will propose to multiple Transaction Agreements (Line~\ref{line:poa-ua-initiate}). By the Termination property of the Transaction Agreement, all instances will terminate and a validator will decide in POA (Line~\ref{line:poa decide slow}).
\end{proof}

\begin{proof}[Fast Termination Property]
    Let's consider some honest validator $V$.
    Given the system is in the synchrony period, by the Common Termination property, $L.X$ will start the instance at most $\Delta$ time after $V.X$ started and hence $L.X$'s proposal $B$ will reach $V.X$ at most $2\Delta$ time after $V.X$ started.
    For every UA-RA transaction $tx \in B$, by assumption, a user who emitted $tx$ is honest, and hence $V$ will successfully FP-lock $tx$ since there is no conflicting transaction and since $V$ will be ready to execute $tx$ by the time $2\Delta$ by \cref{lem:user ready hence validator ready}.
    Also, by \cref{lem:transaction agreement}, for every RA-RA transaction $tx \in B$ $V$ will emit $tx$ at most $\Delta$ time after $L$ did.
    Hence $V.X$ will issue a vote for $B$. 
	
    Therefore, due to synchrony, and by the assumption that at most $p$ validators misbehave, every honest validator will acquire $n - p$ votes for $B$ within $3\Delta$ time after its own start and will Fast-Path decide $B$. 
\end{proof}

\begin{proof}[Validity Property]
    For simplicity, we only consider the validity of a block decided in the Slow Path, since by the Agreement Property, it's the same block as the one decided in the Fast Path.
	
    Consider a block $B$ decided in the Slow Path.
    Let's first prove condition (I) of the Validity Property:
    If a transaction $tx$ is present in the pool of all honest validators at the start of the protocol and $L$ is honest, then $tx$ is included in $B$.
	
    Consider such $tx$.
    An honest validator $V.X$ proposes to the Quorum Consensus either a block $B_1$ for which $V$ received at least $n - p - 2f$ votes (Line~\ref{line:poa propose to qc for fast path}) or a block $B_2$ consisting of transactions that were present in $n - 2f$ fallback blocks (Line~\ref{line:poa propose to qc fallback}).
    Since $L$ is honest, $B_1$ is its proposal, and it contains $tx$.
    So does $B_2$ because among every $n - f$ fallback blocks there will be at least $n - 2f$ blocks from honest validators and each such block contains $tx$ (by condition to form a Quorum Consensus proposal from fallback blocks at Line~\ref{line:poa get slow path proposal from FB blocks}).
    Therefore, every honest validator's proposal to Quorum Consensus will contain $tx$, and by the (I) condition of a Quorum Validity $tx$ will be in the decided block $B_{qc}$ of Quorum Consensus. If $tx$ is an RA transaction, it will trivially be in the decided block, so assume $tx$ is a UA-RA transaction. After deciding $B_{qc}$, every honest validator will propose $tx$ to the corresponding Transaction Agreement (Line~\ref{line:poa propose ua}), since no conflicting transaction can be SP-Locked since an honest user doesn't issue conflicting transactions. Therefore, $tx$ will be decided in its Transaction Agreement instance by the Quorum Validity, and hence be included in $B$.
	
    Now let's prove condition (II) of Validity, namely that if a transaction is included in the decided block, it was emitted.
    
    Consider some $tx \in B$.
    By condition (II) of Quorum Validity, $tx$ is in the proposal of at least $n - 3f \geq 1$ honest validators, and an honest validator proposes a UA-RA transaction $tx$ to Quorum Consensus only if it SP-locked $tx$ (lines \ref{line:poa sp lock for fp}, \ref{line:poa propose to qc for fast path} or \ref{line:poa sp lock fallback}, \ref{line:poa propose to qc fallback}), which is only possible if $n - 3f \geq 1$ honest validators have $tx$ in their pool which implies (Line~\ref{line:ua-ra transaction unpending}) that this transaction was emitted. And an honest validator proposes an RA transaction only if it was present in $n - 2f$ fallback blocks (Line~\ref{line:poa add ra-ra txs to qc proposal}), meaning at least $n - 3f \geq 1$ honest validators included it in their fallback blocks, meaning they have it in their pool (Line~\ref{line:poa let fallback block}), meaning emitted it (Line~\ref{line:ua-ra ra send}, \cref{alg:ua-ra transaction}).
\end{proof}

\begin{proof}[No Conflict Property]
    For the sake of contradiction, assume that there were POA instances $POA_1$ and $POA_2$ in which blocks $B_1$ and $B_2$ were decided, containing respectively $tx_1$ and $tx_2$, which are conflicting. Now, either one of those blocks were decided in the fast path or none were. 

    Consider the case where w.l.o.g. $B_1$ was decided in the fast path. This means that at least $n - p  - f$ honest validators FP-locked $tx_1$, and hence $B_2$ can receive at most $2f + p < n - p$ votes, hence can not be fast-path decided. Moreover, given $tx_1$ was fast-path decided, by Lemma \ref{lem:huuy} no honest validator can SP-Lock $tx_2$, therefore, no honest validator proposes $tx_2$ to the Transaction Agreement instance number $tx_2.sn$, therefore, by the Quorum Validity of Transaction Agreement, $tx_2$ can not be decided in the Transaction Agreement, and hence in the slow path. 

    Finally, assume that both $B_1$ and $B_2$ were decided in the slow path. But this implies that $tx_1$ and $tx_2$ were both decided in the Transaction Agreement instance number $tx_1.sn$ for user $tx_1.sender$ (which are the same values as $tx_2.sn$ and $tx_2.sender$), which is not possible due to the agreement property of Transaction Agreement. 
\end{proof}

\section{Pseudocode for Transaction Execution}
\label{sec:execution}

\begin{algorithm}[h]
\caption{Transaction Execution (UA-RA and RA-RA)}
\label{alg:ra-tx-execution}
\begin{algorithmic}[1]
    \Upon{event}{\textbf{once} poa.Decide}{k, B} \label{line:ua-ra tx poa decide} \Comment{On $V.X$}
    \State $pool \gets pool \setminus B$
    \ForAll{$tx \in \mathsf{DeterministicOrder}(B)$}
        \If{$tx \in executed$} \label{line:ua-ra executed check}
            \textbf{continue}
        \EndIf
        \State $executed \gets executed \cup \{tx\}$
        \State $\mathsf{effects} \gets \mathsf{VM.Execute}(tx.Code_\mathsf{pre}, tx.consumedObjects)$ \label{line:ua-ra transaction execute code}
        \State $\mathsf{effects} \gets \mathsf{VM.Execute}(tx.Call, \mathsf{effects})$ \label{line:ua-ra transaction execute call}
        \State $\mathsf{effects} \gets \mathsf{VM.Execute}(tx.Code_\mathsf{post}, \mathsf{effects})$
        \ForAll{$raratx \in \mathsf{effects}.raratxs$}
            \If {$raratx.consumedObjects \not\subseteq ownedObjects$}
                \State \textbf{continue}
            \EndIf
            \State $ownedObjects \gets ownedObjects \setminus raratx.consumedObjects$
            \State \textbf{trigger} $\langle\mathrm{il.Send} \mid raratx.receiver, raratx \rangle$ \label{line:ua-ra ra send}
        \EndFor
        \If{$tx$ is UA-RA}
            \State \textbf{trigger} $\langle\mathrm{il.Send} \mid tx.sender, Executed(tx, \textsf{effects}) \rangle$
        \EndIf
    \EndFor
    \State $poa.Initiate(k + 1)$ \label{line:ua-ra tx poa init}
    \EndUpon

    \Statex

    \Upon{event}{il.Deliver}{\langle Executed(tx), \textsf{effects}\rangle, V.X} \Comment{On $V.A$}
    \If{$executed[tx.sn]$}
        \Return
    \EndIf
    \State \textbf{await} $executed[tx.sn-1]$ \textbf{and} $tx.consumedObjects \subseteq ownedObjects$
    \State $ownedObjects  \gets ownedObjects \setminus tx.consumedObjects \cup \textsf{effects}.createdObjects$
    \State $executed[tx.sn] \gets \textbf{true}$
    \EndUpon
\end{algorithmic}
\end{algorithm}

\section{System Analysis}
\label{sec:proofs-transaction-execution}

\subsection{Formal Correctness Properties}

Properties for reactive actors follow the schema of Total Order Broadcast~\cite{defago2004total}.

\begin{property}[Reactive Actor Agreement]
    For any honest validators $V_1$ and $V_2$, a reactive actor $X$ and a transaction $tx$, if $V_1.X$ executes $tx$, then $V_2.X$ eventually executes $tx$.
\end{property}

\begin{property}[Reactive Actor Validity]
    If an honest user emits a UA-RA transaction, it is eventually executed by all correct validators.

    If a reactive actor on any honest validator emits an RA-RA transaction, it is eventually executed by all correct validators.
\end{property}

\begin{property}[Reactive Actor Total Order]
    For any reactive actor $X$, honest validators $V_1, V_2$, and transactions $tx_1, tx_2$, if $V_1.X$ executes $tx_1$ before $tx_2$, then $V_2.X$ executes $tx_1$ before $tx_2$.
\end{property}

\begin{property}[Reactive Actor Integrity]
    For any reactive actor $X$, honest validator $V$ and transaction $tx$, $V.X$ executes $tx$ at most once and only if $tx$ was emitted.
\end{property}

User actor properties follow the schema of Byzantine Reliable Broadcast~\cite{brb}.

\begin{property}[User Actor Validity]
    If a correct user emits a user actor transaction, this transaction is eventually executed by every honest validator.
\end{property}

\begin{property}[User Actor Integrity]
    Every correct validator executes each user actor transaction at most once and only if it was emitted.
\end{property}

\begin{property}[User Actor Agreement]
    For any user actor $A$, honest validators $V_1, V_2$ and user actor transactions $tx_1, tx_2$ both with sender $A$ and the same sequence number, if $V_1.A$ executes $tx_1$ and $V_2.A$ executes $tx_2$, then $tx_1 = tx_2$. 
\end{property}

\subsection{Formal Fast Execution Properties}
\begin{property}[Fast UA Transaction Execution]
    Given an honest user $A$ issues a UA transaction $tx$, the system is after GST, and at most $p$ validators are faulty, $tx$ is executed after $2$ communication steps. 
\end{property}

\begin{property}[Fast UA-RA Transaction Execution]
    Given an honest user $A$ issues a UA-RA transaction $tx$ to a reactive actor $X$, the system is after GST, the leader $L$ of the next POA instance of $X$ is honest and starts the instance as soon as it receives a transaction from $A$, and at most $p$ validators misbehave, then $tx$ is executed after $3$ communication steps.
\end{property}

\begin{property}[Fast RA-RA Transaction Execution]
    Given a reactive actor issues an RA transaction $tx$ to a reactive actor $Y$, the system is after GST, the leader $L$ of the next POA instance of $Y$ is honest and starts the instance as soon as it receives a transaction from $X$, and at most $p$ validators misbehave, then $tx$ is executed after $2$ communication steps.
\end{property}

\subsection{Proofs for All System Properties}

\begin{proof}[Reactive Actor Agreement]
    If $V_1.X$ executes $tx$, it means that it decided a block $B$ containing $tx$.
    Assume $V_1.X$ decided $B$ in the $k$-th instance of POA.
    By the Multi-Termination and Agreement properties of POA, $V_2.X$ will eventually decide in the $k$-th instance of POA for $X$ and its decision will be $B$.
    Therefore, $V_2.X$ will also eventually execute $tx$.
\end{proof}

\begin{corollary}[Reactive Actor Agreement]
\label{cor:eventual-emit}
    For two honest validators $V_1$ and $V_2$, and a reactive actor $X$, if $V_1.X$ emits $tx$, then $V_2.X$ eventually emits $tx$.
\end{corollary}
\begin{proof}
    Let $tx'$ be a transaction executing a $Call$ of which made $V_1.X$ emit $tx$.
    By the Reactive Actor Agreement property, $V_2.X$ will eventually execute $tx'$ having the same state and the same set of owned objects as $V_1.X$ had when executing $tx'$.
    Therefore, $V_2.X$ will also emit $tx$. 
\end{proof}

\begin{remark}
    \cref{def:emit} is independent of a validator since if a transaction was emitted by a reactive actor of one validator, by \cref{cor:eventual-emit} of the Reactive Actor Agreement property, every honest validator will eventually emit it. 
\end{remark}

\begin{definition}
    Consider an execution $E$ and two honest validators $V_1$ and $V_2$.
    We define a \emph{direct ancestor according to $V_1$} relation for two transactions $tx_1$ and $tx_2$ executed by $V_1$ in $E$ the following way: $tx_1$ is a direct ancestor of $tx_2$ if either (i) $tx_1$ and $tx_2$ are emitted by the same user actor and $tx_1.sn = tx_2.sn - 1$ or (ii) $tx_2$ consumes objects created by $tx_1$. Define an \emph{ancestor} relation as a transitive closure of a direct ancestor relation. 
\end{definition}

\begin{lemma}\label{lem:ancestors hence current}
    Consider two honest validators $V_1$ and $V_2$ and execution $E$.
    If $V_1$ executes some transaction $tx$ at time $t$ in $E$ and $V_2$ executes all ancestors of $tx$ according to $V_1$ by at most $\max(t + \Delta, GST + \Delta)$ then $V_2$ executes $tx$ by at most $\max(t + \Delta, GST + \Delta)$. 
\end{lemma}
\begin{proof}
    In case $tx$ is either UA-RA, RA-RA or RA, this follows from the Reactive Actor Agreement and Common Termination properties without a need for the ancestor premise.
    Therefore, we focus on the case of $tx$ being a UA transaction.
    Denote $A := tx.sender$ and $tx'$ a transaction with $sender = A$ and $tx'.sn = tx.sn - 1$.
    Note that if $V_1$ executed $tx$, it means it decided it in Parallel Optimistic Broadcast (Line~\ref{line:pob decide}, \cref{alg:ua transaction}), meaning $V_2$ will also eventually decide it in Parallel Optimistic Broadcast and will put it into $pending$ (Line~\ref{line:ua tx put in pending}).
    And $V_2$ will be ready to execute $tx$ (Line~\ref{line:ua unpending}) by the time of at most  $\max(t + \Delta, GST + \Delta)$ since by the state of the Lemma, by that time $V_2$ will execute all ancestors of $tx$. 
\end{proof}

\begin{lemma}
\label{lem:transaction agreement}
    If an honest validator executes a transaction $tx$ at time $t$, then every honest validator executes $tx$ at time at most $\max(t + \Delta, GST + \Delta)$.
\end{lemma}
\begin{proof}
    For UA-RA and RA-RA transactions, this follows from the Reactive Actor Agreement and Common Termination properties. 
    We now give proof for UA transactions. 
	
    For the sake of contradiction, consider an execution $E$ of the protocol, two honest validators $V_1$ and $V_2$, and a UA transaction $tx$ such that in $E$ $V_1$ executes $tx$ at time $t$ and $V_2$ does not execute $tx$ by $\max(t + \Delta, GST + \Delta)$.
	
    By \cref{lem:ancestors hence current}, if $V_1$ executes $tx$ but $V_2$ does not, it means that there is a transaction that is an ancestor of $tx$ according to $V_1$ that is not executed by $V_2$.
    Repeat the proof process for that ancestor.
    Since there is a finite number of ancestors of each transaction and an ancestor relation is transitive, we end up with an ancestor of $tx$ that is not executed by the time $\max(t + \Delta, GST + \Delta)$ but all its ancestors are, which contradicts \cref{lem:ancestors hence current}.
\end{proof}

\begin{lemma}
\label{lem:user ready hence validator ready}
    Consider a user actor $A$ controlled by honest user $U$ and a transaction $tx$ from $A$.

    If $U$ emits $tx$ at time $t$, then for every honest validator $V$, $V.A$ will have $executed[tx.sn - 1]$ and $tx.objects \subseteq ownedObjects$ by $\max(t + \Delta, GST + \Delta)$.
\end{lemma}
\begin{proof}
    Since $U$ is honest and emits $tx$, we conclude that $U$ saw effect of some transactions $tx_1, \ldots, tx_k$ that made $executed[A][tx.sn]$ and $tx.objects \subseteq objects[A]$ hold.
    It means that $\forall i \in [k]$ $U$ received at least $f + 1$ votes for $tx_i$, meaning there is at least one honest process which executed $tx_i$, which by \cref{lem:transaction agreement} implies that every honest validator will execute $tx_i$ by $\max(t + \Delta, GST + \Delta)$.
    Therefore, for every honest validator $V.A$ $executed[tx.sn - 1]$ and $tx.objects \subseteq ownedObjects$ will hold by $\max(t + \Delta, GST + \Delta)$.
\end{proof}

\begin{proof}[Reactive Actor Validity]
    First, consider an honest user $U$ issues a UA-RA transaction $tx$ to a reactive actor $X$.
    Since $U$ is honest, checks in Line~\ref{line:ua-ra checks} of \cref{alg:ua-ra transaction} will pass and each honest validator will put $tx$ into its $pending$ set.
    Moreover, by \cref{lem:user ready hence validator ready}, a condition in Line~\ref{line:ua-ra transaction unpending}, \cref{alg:ua-ra transaction} will eventually hold for every honest $V.A$ and it will send $tx$ to $V.X$ (Line~\ref{line:ua-ra transaction send}).
    So, $tx$ will eventually end up in the pool of every honest validator (Line~\ref{line:ua-ra transaction pool}), let's denote this time point with $t$.
    Now consider a point in time after $t$ and after GST when an honest validator $L$ becomes a leader of POA for $X$.
    If by this time $L$ already executed $tx$, then by the Reactive Actor Agreement $tx$ will be eventually executed by every honest validator.
    Otherwise, $L$ will include $tx$ in its proposal and by the (I) part of the Validity property of POA, $tx$ will be decided and then executed (Lines~\ref{line:ua-ra transaction execute code} and \ref{line:ua-ra transaction execute call}).

    Now, we give a prove for an RA-RA transaction from a reactive actor $X$ to a reactive actor $Y$.
    \cref{cor:eventual-emit} states that if for some honest validator $V_1$, $V_1.X$ emitted $tx$, then for every other honest validator $V_2$, $V_2.X$ will eventually emit $tx$.
    Therefore, $tx$ will end up in the pool of $V.Y$ for every honest validator $V$, and by the same logic as with a UA-RA transaction, $tx$ will be eventually executed by every honest validator. 
\end{proof}

\begin{proof}[Reactive Actor Total Order]
    A reactive actor validator executes a transaction only if it decides a block in POA containing this transaction. 
    
    Let $B_1$ be the block containing $tx_1$ and $k_1$ be the number of POA instance in which $B_1$ is decided.
    Analogously we define $B_2$ and $k_2$ for $tx_2$. 
	
    Given $V_1.X$ executed $tx_1$ before $tx_2$ and since processes execute blocks sequentially, we conclude that $k_1 \leq k_2$.
    If now $k_1 < k_2$, then $V_2.X$ executes $tx_1$ before $tx_2$ because of the sequential execution of blocks.
    And if $k_1 = k_2$, then $B_1 = B_2$, and this block is executed in the same order by $V_1.X$ and $V_2.X$ due to the deterministic block execution. 
\end{proof}

\begin{proof}[Reactive Actor Integrity]
    If $tx$ is executed by $V.X$, it means that $V.X$ decided a block containing $tx$.
    $V.X$ can decide either on the Fast Path (Line~\ref{line:poa-decide-optimistic}, \cref{alg:poa}) or on the slow path (Line~\ref{line:poa decide slow}, \cref{alg:poa}). 
	
    If a block containing $tx$ is decided on the fast path, it means it received at least $n - p - f \geq 1$ honest votes (Line~\ref{line:poa-fast-path-condition}), meaning it passed the checks of an honest validator.
    In case $tx$ is a UA-RA transaction, this means that $tx$ is correctly signed by a user (Line~\ref{line:poa check user sign}) and a validator successfully FP-locked $tx$ (Line~\ref{line:poa try fp-lock}), meaning it has $executed[tx.sn - 1] = true$ and $tx.objects \subseteq ownedObjects$ (Line~\ref{line:poa approve fp lock condtion}), therefore $tx$ is emitted.
    In case $tx$ is an RA-RA transaction, it means that it is contained in the pool of an honest validator (Line~\ref{line:poa receive leader's block condition}), hence it was issued by a reactive actor of an honest validator, and hence it is emitted.
	
    If $tx$ is included in the block decided on the slow path, then by the (II) part of Quorum Validity, there are $n - 3f \geq 1$ honest validators who proposed a block containing $tx$ to the Quorum Consensus.
    An honest validator includes $tx$ in its proposal either if it received $n - p - 2f$ votes for the block containing $tx$ (Line~\ref{line:poa propose to qc for fast path}) or if it received $n - 2f$ fallback blocks containing $tx$ (Line~\ref{line:poa propose to qc fallback}).
    In case it received $n -  p - 2f$ votes, there were at least $n - p - 3f \geq p + 1$ honest votes for a block containing $tx$, which by the argument above implies that $tx$ is emitted.
    In case it received $n - 2f$ fallback blocks containing $tx$, it means that at least $n - 3f \geq 1$ honest validators included $tx$ into their fallback block, meaning $tx$ was in their pool, meaning it was either sent by a reactive actor (Line~\ref{line:ua-ra ra send}, \cref{alg:ua-ra transaction}) and thus emitted, or it was sent by a user actor (Line~\ref{line:ua-ra transaction send}, \cref{alg:ua-ra transaction}) and hence checked for a correct user signature (Line~\ref{line:ua-ra checks}) and for $executed[tx.sn - 1] = true$ and $tx.objects \subseteq ownedObjects$ (Line~\ref{line:ua-ra transaction unpending}) and thus emitted.

    Every transaction is executed at most once by a reactive actor since it maintains an $executed$ set and skips a transaction if it was executed already (Line~\ref{line:ua-ra executed check}, \cref{alg:ua-ra transaction}).   
\end{proof}

\begin{proof}[User Actor Validity]
    Consider an emitted transaction $tx$ issued by user $A$ with $tx.sender = A$.
    Given that $A$ is honest, by the Validity Property of Parallel Optimistic Broadcast, $tx$ will eventually be decided in the Parallel Optimistic Broadcast and will end up in a $pending$ set of every honest validator (Line~\ref{line:ua tx to pendig}, \cref{alg:ua transaction}).
    And by \cref{lem:user ready hence validator ready}, for every honest $V.A$ $executed[tx.sn - 1]$ and $tx.objects \subseteq ownedObjects$ will eventually hold (Line~\ref{line:ua unpending}) and thus $V.A$ will execute $tx$ (Line~\ref{line:ua tx execute}).
\end{proof}

\begin{proof}[User Actor Integrity]
    A correct validator executes a transaction only if it previously put it into a $pending$ set (Line~\ref{line:ua unpending}, \cref{alg:ua transaction}).
    A correct validator puts a transaction into a $pending$ set only if it delivers it in a Parallel Optimistic Broadcast instance (Line~\ref{line:ua tx put in pending}).
    By the Integrity property of the Parallel Optimistic Broadcast, at most one transaction with a given sequence number will be decided in the given instance, and hence, at most one transaction with a given sequence number will be executed.

    Furthermore, a transaction is put in the $pending$ set only if it is correctly signed and is executed only if the condition in Line~\ref{line:ua unpending} holds, meaning only if a transaction is emitted. 
\end{proof}

\begin{proof}[User Actor Agreement]
    We consider all possible cases for $tx_1$ and $tx_2$ to be either UA or UA-RA transactions. 

    In case both $tx_1$ and $tx_2$ are UA transactions, we can conclude that $V_1.A$ and $V_2.A$ delivered them in the same Parallel Optimistic Broadcast instance, hence $tx_1 = tx_2$ by the Agreement property of Parallel Optimistic Broadcast. 

    If $tx_1$ and $tx_2$ are both UA-RA transactions, we conclude that they were both decided in the POA and hence, due to the No Conflict property of POA, $tx_1 = tx_2$.

    Finally, it can not be that $tx_1$ is a UA transaction, and $tx_2$ is a UA-RA transaction. That is because, given Agreement properties of Parallel Optimistic Broadcast and POA, we can assume that both $tx_1$ and $tx_2$ were decided in the slow path, that is in the instance number $tx_1.sn$ (same as $tx_2.sn$) of Transaction Agreement of $A$. But then $tx_1 = tx_2$ by the Agreement property of Transaction Agreement.
    
\end{proof}

\begin{proof}[Fast UA Transaction Execution]
    Each honest validator will terminate the Byzantine Reliable Broadcast of $tx$ in two communication steps.
    And by \cref{lem:user ready hence validator ready}, every honest validator will be ready to execute $tx$ by that time.
\end{proof}

\begin{proof}[Fast UA-RA Transaction Execution]
    The first communication step is a user's broadcast of $tx$, so, in particular, $L$ receives it.
    Then we apply the Fast Termination property of POA which adds two additional communication steps.
\end{proof}





\begin{proof}[Multi-termination]
    Due to the Termination property of POA, it's sufficient to show that $V.X$ will initiate $k$-th instance. This we prove by induction on $k$. 

    The first POA instance is initiated immediately when the reactive actor entity is initialized.
    Now, assume a validator decides in the $k-1$ instance.
    This would trigger a $poa.Decide(k-1, \cdot)$ event (Line~\ref{line:ua-ra tx poa decide}, \cref{alg:ua-ra transaction}) and thus a $poa.Initiate(k)$ event (Line~\ref{line:ua-ra tx poa init}).
\end{proof}

\end{document}